\documentclass[journal]{IEEEtran}
\usepackage{graphics}
\usepackage{graphicx}
\usepackage{epsfig}
\usepackage{amsmath,amssymb}
\usepackage{mathrsfs,amsfonts,fancyhdr}
\usepackage{multirow}
\usepackage{cite}
\usepackage{amsthm}
\newtheorem{thm}{Theorem}[]
\newtheorem{cor}{Corollary}[]

\theoremstyle{remark}

\begin{document}
\title{Rate Maximization for Half-Duplex Multiple Access with Cooperating Transmitters}
\author{\IEEEauthorblockN{Ahmad Abu Al Haija,~\IEEEmembership{Student Member,~IEEE,} and Mai Vu,~\IEEEmembership{Member,~IEEE}}
\vspace{-5 mm}
\thanks{This work has been supported in part by grants from the Natural Science and Engineering Research Council of Canada (NSERC) and the Fonds Quebecois de la
Recherche sur la Nature et les Technologies (FQRNT). A part of this work has been presented at the IEEE International Conference on Communications (ICC) $2011$.

Ahmad Abu Al Haija is with the Department of Electrical and Computer Engineering, McGill University, Montreal, Canada
(e-mail: ahmad.abualhaija@mail.mcgill.ca). He is currently visiting Tufts University.
Mai Vu is with the Department of Electrical and Computer Engineering, Tufts University, Medford, MA, USA
(e-mail: maivu@ece.tufts.edu).}}
\maketitle
\begin{abstract}
 We derive the optimal resource allocation of a practical half-duplex scheme for the Gaussian multiple access channel with transmitter cooperation (MAC-TC). Based on rate splitting and superposition coding, two users transmit information to a destination over
 $3$ phases, such that the users partially exchange their information during the first $2$ phases and cooperatively transmit
 to the destination during the last one. This scheme is near capacity-achieving when the inter-user links are stronger than each user-destination link; it also includes partial decode-forward relaying as a special case. We propose efficient algorithms to find the optimal resource allocation for
 maximizing either the individual or the sum rate and identify the corresponding optimal scheme for each channel configuration. For fixed phase durations, the power allocation problem is convex and can be solved analytically based on the KKT conditions. The optimal phase durations can then be obtained numerically using simple search methods.  Results show that as the inter-user link qualities increase, the optimal scheme moves from no cooperation to partial then to full cooperation, in which the users fully exchange their information and cooperatively send it to the
 destination. Therefore, in practical systems with strong inter-user links, simple decode-forward relaying at both users is rate-optimal.
\end{abstract}
\section{Introduction}\label{sec:intro}
\IEEEPARstart{C}{ooperation} among nodes in wired or wireless networks can significantly improve network throughput and reliability \cite{Erkip, will1, gama}. A large number of cooperative schemes have been proposed for fundamental networks such as the relay channel and the multiple access channel with transmitter cooperation (MAC-TC) \cite{will1, gama, Erkip, HoMa, HoMa2, haivu3}. The MAC-TC is particularly interesting as it includes the classical MAC and also the relay channel as special cases. Furthermore, it has immediate applications in cellular and ad hoc networks. For example, in the uplink of a cellular system, two mobiles can cooperate to send their information to the base station. This cooperation leads to a larger rate region and smaller outage probability as shown in \cite{Erkip}. In ad hoc or sensor networks,  two or more  nodes with good inter-node link qualities can also cooperate to send their information to a common destination.

An important question from the practical implementation perspective is the optimal resource allocation that achieves the maximum performance
in a channel. Optimal resource allocation has been studied quite extensively for non-cooperative channels.
In \cite{itwt}, for example, iterative water-filling is proposed to maximize the sum capacity of the Gaussian MIMO MAC. Optimal
power allocations for minimizing the outage probability  of the fading MAC and broadcast channel are derived in \cite{gold1} and \cite{FBC}, respectively.
These works
are based on formulating the problems as a convex optimization and then solving the KKT conditions.

Recently, attention has been focused on resource allocation for cooperative communications. For the half-duplex relay channel, power allocations for maximizing achievable rates of various schemes and an upper bound are derived in \cite{HoMa, HoMa2}, and for minimizing the power consumption over AWGN channels in \cite{Fanny1}. For the full-duplex relay channel, power allocation for minimizing the outage probability using multi-hop or decode-forward relaying is derived in \cite{alou1}. For the full-duplex MAC-TC with the generalized feedback scheme analyzed in \cite{Erkip}, the optimal power allocation is derived in \cite{fhmac}. Moreover, reference \cite{fhmac} also puts forward a new half-duplex scheme with its optimal power allocation. This half-duplex scheme, however,
is sub-optimal as it is not based on sound information-theoretic analysis.

In \cite{haivu3}, we have proposed a near capacity-achieving half-duplex cooperative scheme for the MAC-TC consisting of two users communicating with a destination. The scheme is based on rate splitting, superposition coding and partial decode-forward relaying. The transmission occurs over independent blocks, each of which is divided into three phases. During the first two phases, the two users exchange part of their information; then during the last phase, they cooperatively transmit their information to the destination. This scheme is near capacity-achieving as shown in \cite{haivu3}, especially when the inter-user link qualities are higher than the link quality
between each user and the destination. Moreover, it includes as a special case the half-duplex decode-forward scheme for relay channels as proposed
in \cite{HoMa, HoMa2}.

In this paper, we derive the optimal power allocation and phase durations that maximize the individual or sum rate of this scheme applied in the Gaussian channel with Gaussian signaling. This optimization generalizes our previous result in \cite{haija1} for the symmetric channel to the general asymmetric case. Since the considered problem for either individual or sum rate maximization is convex only with fixed phase durations, we decompose the problem into $2$ steps \cite{dec2006}. First, we fix the phase durations and derive the optimal power allocation by analyzing the KKT conditions of the obtained convex problem. Then, we find the optimal phase durations numerically. Depending on the link qualities between two users and the destination, we analyze different cases resulting from the KKT conditions, derive the power allocation and obtain the corresponding optimal  scheme for each case.

Analysis and numerical results show that a user chooses to cooperate with the other if the link to this user is stronger than to the destination, but abstains from cooperation if this link is weaker. As the inter-user link quality increases, the
amount of information sent cooperatively
also increases, such that the scheme will transverse from partial to full cooperation.
 Using optimization results, we
 further characterize the optimal scheme for each geometric location of the destination on a $2D$ plane with respect to the locations of the two users where the channel gain between any two nodes is related to the distance
by a pathloss-only model. We present the geometric optimal regions of each scheme. These analyses can be useful for network planners in obtaining
the optimal performance.
\section{Channel Model}\label{sec:system_model}
The Gaussian MAC with transmitter cooperation (MAC-TC)
consists of two users and one destination. The communication links among these
terminals are characterized by complex-valued channel gains and additive while Gaussian noise (AWGN).  The mathematical formulation of this channel
can be expressed as \cite{Erkip}
\begin{align}\label{chan1}
Y_{12}&=h_{12}X_1+Z_1,\;\;\nonumber\\
Y_{21}&=h_{21}X_2+Z_2,\;\;\nonumber\\
Y_{3}&=h_{10}X_1+h_{20}X_2+Z_3,
\end{align}
\noindent where $X_1$ and $X_2$ are signals transmitted from user $1$ and user $2$, respectively; $Y_{21},$ $Y_{12},$ and $Y_3$
are the signals received by user $1$, user $2$, and the destination, respectively; $Z_1, Z_2,$ and $Z_3$ are i.i.d complex Gaussian noises with zero mean and unit variance; $h_{12},$ and $h_{21}$ are
inter-user link coefficients, whereas $h_{10}$ and $h_{20}$ are the
link coefficients between the users and destination. Each link coefficient is a complex value $h_{ij}=g_{ij}e^{\sqrt{-1}\theta_{ij}}$ where $g_{ij}$ is the real amplitude gain and  $\theta_{ij}$ is the phase.
We assume that each user knows all link amplitudes and phases of all links to the
destination. This information can be obtained through feedback from the destination or channel reciprocity
as discussed in \cite{Erkip}. We also assume that  each receiver can compensate for the phases perfectly.

The half-duplex constraint is satisfied
by using time division consisting of $3$ phases as shown in Figure \ref{fig:Gausmod}.
The channel model in each phase  can be expressed as
\noindent
\begin{align}\label{recsig}
\text{phase}\;1:&\;\;\;Y_{12}=h_{12}X_{11}+Z_1,\;\;Y_{1}=h_{10}X_{11}+Z_{31},\\
\text{phase}\;2:&\;\;\;Y_{21}=h_{21}X_{22}+Z_2,\;\;Y_{2}=h_{20}X_{22}+Z_{32},\nonumber\\
\text{phase}\;3:&\;\;\;\;Y_{3}=h_{10}X_{13}+h_{20}X_{23}+Z_{33},\nonumber
\end{align}
\noindent where $Y_{ij},\;(i,j)\in\{1,2\}$, is the signal received by the $j^{\text{th}}$ user
during the $i^\text{th}$ phase;  $Y_k,\;k\in\{1,2,3\}$ is the signal received by the destination during the $k^\text{th}$ phase; and all the $Z_l,\;l\in\{1,2,31,32,33\}$, are i.i.d complex Gaussian noises with zero mean and unit variance. $X_{11}$ and $X_{13}$ are the signals transmitted from
 user $1$ during the $1^{\text{st}}$ and $3^{\text{rd}}$ phases, respectively. Similarly, $X_{22}$ and $X_{23}$ are the signals transmitted from user $2$
 during the $2^{\text{nd}}$ and $3^{\text{rd}}$ phases.
\begin{figure}[t]
    \begin{center}
    \includegraphics[width=90mm]{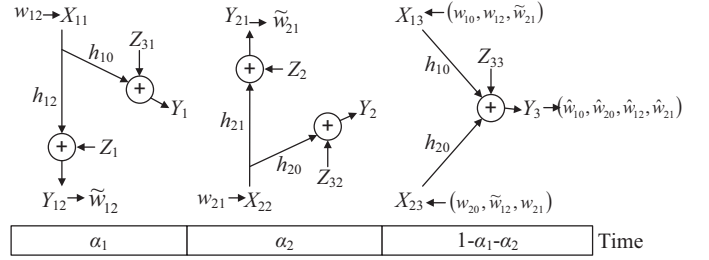}
    \caption{Half-duplex coding scheme for the Gaussian MAC-TC.} \label{fig:Gausmod}
    \end{center}
\vspace*{-6mm}
\end{figure}
\section{Cooperative Scheme and Rate Region}\label{sec:copshrr}
In this section, we  describe the proposed half-duplex scheme for the MAC-TC. We then show its achievable rate region and analyze the
special case of partial decode-forward relaying.
\vspace{-4 mm}
\subsection{Cooperative Scheme}
 We use a block coding scheme in which every block has the same length and is coded independently. Using time division, each block is divided into $3$ phases with durations  $\alpha_1,$ $\alpha_2$ and $\alpha_3=1-\alpha_1-\alpha_2$. Let
$w_1$ and $w_2$ be the messages to be sent during a specific
block by user $1$ and user $2$, respectively. Each user splits its message into two
parts: $w_1=(w_{10}, w_{12})$ and $w_2=(w_{20}, w_{21})$, where $w_{12}$ and $w_{21}$ are cooperative message parts to be decoded by the other user and the destination, $w_{10}$ and $w_{20}$ are private message parts to be decoded by the destination only. During the
$1^\text{st}$ phase, user $1$ sends $w_{12}$ and
 user $2$ decodes it. Similarly, during the
$2^\text{nd}$ phase, user $2$ sends $w_{21}$ and
the user $1$ decodes it. Then, during the $3^\text{rd}$ phase, user $1$ sends both cooperative message parts and its own private part as $(w_{10},w_{12},w_{21})$. Similarly,
user $2$ sends $(w_{20},w_{21},w_{12})$. The users encode these messages by superposition coding.
At the end of the $3^\text{rd}$ phase, the destination utilizes signals received in all three phases to
jointly decode both messages $(w_1,w_2)$. Figure \ref{fig:Gausmod} illustrates this proposed scheme.
\subsubsection{Signaling}

User $1$ first generates a codeword for its cooperative message part in the $1^{\text{st}}$ phase and constructs the signal $X_{11}$. It then superimposes its private message part on both cooperative message parts and constructs the signal $X_{13}$. Similarly for user $2$. Therefore, both users construct their transmit signals as
\noindent
\begin{align}\label{sigtr}
X_{11}&=\sqrt{\rho_{11}}U_{1}(w_{12}),\;\;X_{22}=\sqrt{\rho_{22}}U_{2}(w_{21})\nonumber\\
X_{13}&=\sqrt{\rho_{10}}V_{1}(w_{10})+\sqrt{\rho_{13}}U_3(w_{12},w_{21}),\nonumber\\
X_{23}&=\sqrt{\rho_{20}}V_{2}(w_{20})+\sqrt{\rho_{23}}U_3(w_{12},w_{21})
\end{align}
\noindent where $U_{1},U_{2},V_{1},V_{2}$ and $U_3$ are independent and identically distributed Gaussian codewords with entries that have zero mean and unit variance. Here, $\rho_{11},\rho_{22},\rho_{10}$ and $\rho_{20}$ are the transmission powers allocated for codewords $U_{1},U_{2},V_{1}$ and $V_{2}$, respectively,
$\rho_{13}$ and $\rho_{23}$ are the transmission powers allocated for codeword $U_3$ by user $1$ and user $2$, respectively. These power allocations satisfy the following
power constraints:
\noindent
\begin{align}\label{powcsch2}
\alpha_1\rho_{11}+\alpha_3(\rho_{10}+\rho_{13})&=P_1,\nonumber\\
\alpha_2\rho_{22}+\alpha_3(\rho_{20}+\rho_{23})&=P_2.
\end{align}
Then, in the $1^{\text{st}}$ phase, user $1$ sends $w_{12}$ at rate $R_{12}$  by transmitting the signal $X_{11}$ and user $2$ decodes $w_{12}$.
Similarly, in the $2^{\text{nd}}$ phase, user $2$ sends $w_{21}$ at  rate $R_{21}$  by transmitting the signal $X_{22}$ and user $1$ decodes $w_{21}$.
Finally, during the last phase, user $1$ sends $(w_{10},w_{12},w_{21})$ at  rate triplet $(R_{10},R_{12},R_{21})$ by transmitting $X_{13}$. Similarly for user $2$. Since both users
know $w_{12}$ and $w_{21}$ in this phase, they can perform coherent transmission of these cooperative message parts by transmit beamforming such that the achievable
rates of both users are increased.
\subsubsection{Decoding}
\subsubsection*{At each user}
In the $1^{\text{st}}$ phase, user $2$ decodes $w_{12}$ from $Y_{12}$ using maximum likelihood (ML) decoding. User $2$ can reliably decode $w_{12}$  if
\begin{align}\label{spr12}
R_{12}\leq&\; \alpha_1\log\left(1+g_{12}^2\rho_{11}\right)=I_1.
\end{align}
Similarly, in the $2^{\text{nd}}$ phase, user $1$ applies ML decoding and can reliably decode $w_{21}$  if
\begin{align}\label{spr21}
R_{21}\leq&\; \alpha_2\log\!\left(1+g_{21}^2\rho_{22}\right)=I_2.
\end{align}
\subsubsection*{At the destination}
The destination utilizes the received signals in all three phases $(Y_1,Y_2,Y_3)$ to
jointly decode all message parts $(w_{10},w_{20},w_{12},w_{21})$ using joint ML decoding. Specifically, it looks for a  message vector
$(\hat{w}_{10},\hat{w}_{20},\hat{w}_{12},\hat{w}_{21})$ such that the conditional probability of the received signal vector $(Y_1, Y_2, Y_3)$ as in (\ref{recsig}) given
the codewords in (\ref{sigtr}) are the highest among all other message vectors. The destination can reliably decode this message vector if
\begin{align}\label{sprd}
R_{10}\leq&\; \alpha_3\log\left(1+g_{10}^2\rho_{10}\right)=I_3\\
R_{20}\leq&\; \alpha_3\log\left(1+g_{20}^2\rho_{20}\right)=I_4\nonumber\\
R_{10}+R_{20}\leq&\;\alpha_3\log\left(1+g_{10}^2\rho_{10}+g_{20}^2\rho_{20}\right)=I_5\nonumber\\
R_1+R_{20}\leq&\; \alpha_1\log\left(1+g_{10}^2\rho_{11}\right)+\alpha_3\log(1+
\zeta)=I_6\nonumber\\
R_{10}+R_2\leq&\;\alpha_2\log\left(1+g_{20}^2\rho_{22}\right)
+\alpha_3\log(1+\zeta)=I_7\nonumber\\
R_1+R_2\leq&\; \alpha_1\log\left(1+g_{10}^2\rho_{11}\right)+\alpha_2\log\left(1+g_{20}^2\rho_{22}\right)\nonumber\\
&+\alpha_3\log(1+\zeta)=I_8.\nonumber
\end{align}
where
\begin{align}\label{jowz}
\!\!\!\zeta&=g_{10}^2(\rho_{10}+\rho_{13})+g_{20}^2(\rho_{20}+\rho_{23})+2g_{10}g_{20}\sqrt{\rho_{13}\rho_{23}}
\end{align}
These constraints are obtained as follows: $I_j$ is the maximum reliable transmission rate for a message $w_{d}$ such that the destination can decode $w_{d}$ reliably  if it has decoded some messages  $w_{c}$ correctly, where for
  \begin{itemize}
    \item $j=3$, $w_d=w_{10}$ and $w_c=(w_{20},w_{12},w_{21})$
    \item $j=4$, $w_d=w_{20}$ and $w_c=(w_{10},w_{12},w_{21})$
    \item $j=5$, $w_d=(w_{10},w_{20})$ and $w_c=(w_{12},w_{21})$
    \item $j=6$, $w_d=(w_{10},w_{20},w_{12})$ and $w_c=w_{21}$
    \item $j=7$, $w_d=(w_{10},w_{20},w_{21})$ and $w_c=w_{12}$
    \item $j=8$, $w_d=(w_{10},w_{20},w_{12},w_{21})$ and $w_c=\emptyset$.
  \end{itemize}
  The decoding for $w_{d}$ given $w_c$ can be implemented as maximum likelihood decoding jointly among all messages in $w_d$.
Note that the terms $I_6, I_7$ and $I_8$ show the advantage of beamforming resulted from coherent transmission of $(w_{12},w_{21})$ from both users in the
$3^{\text{rd}}$ phase.

\vspace{-3 mm}
\subsection{Achievable Rate Region}
The achievable rate region in terms of $R_1=R_{10}+R_{12}$ and $R_2=R_{20}+R_{21}$ is given in the following theorem:
\vspace{-2 mm}
\begin{thm}
The achievable rate region resulting from the proposed scheme consists of rate pairs $(R_1,R_2)$ satisfying the following constraints:
\begin{align}\label{th1Grrsimp}
R_1\leq&\; I_1+I_3\triangleq J_1\\
R_2\leq&\; I_2+I_4\triangleq J_2\nonumber\\
R_1+R_2\leq&\; I_1+I_2+I_5\triangleq S_1\nonumber\\
R_1+R_2\leq&\; I_2+I_6\triangleq S_2\nonumber\\
R_1+R_2\leq&\; I_1+I_7\triangleq S_3\nonumber\\
R_1+R_2\leq&\; I_8\triangleq S_4\nonumber
\end{align}
\noindent for some $\alpha_1\geq0,\;\alpha_2\geq 0,$ $\alpha_1+\alpha_2\leq1$ and power allocation set $(\rho_{10},\rho_{20},\rho_{11},\rho_{22},\rho_{13},\rho_{23})$ satisfying constraints in (\ref{powcsch2}).
\end{thm}
\begin{proof}
Combining (\ref{spr12}), (\ref{spr21}) and (\ref{sprd}) leads to (\ref{th1Grrsimp}).
\end{proof}

As a special case of the proposed scheme, if one user has no information to send and just relays the cooperative message part of the other user to the destination, the scheme becomes a partial decode-forward (PDF) scheme for the half-duplex relay channel. For example, if user $2$ has no information to send, then $\alpha_2=0,$ $\rho_{22}=\rho_{20}=0$ and $\rho_{23}=(1-\alpha_1)^{-1}P_2$. The achievable rate for user $1$ $(R_1)$ in this case is given as in (\ref{th1Grrsimp}) with $J_2=0,$ $J_1=S_1,$
$S_2=S_4$ and $S_3$ redundant since $S_3>S_1$. Therefore, for the half-duplex relay channel, the achievable rate in (\ref{th1Grrsimp})  becomes $R_1\leq \min\{J_1,S_4\}$.

In the original PDF scheme for the full-duplex relay channel \cite{hsce6120}, the source sends both message parts and the relay decodes
 only one. In \cite{Fanny1}, this partial decode-forward relaying is adapted to $2$ phases for the half-duplex relay channel. In \cite{HoMa, HoMa2}, a slightly different  two-phase
 scheme  is proposed where only one message part is transmitted in the first phase. Our scheme also transmits only one message part in
 each of the first two phases; however, it is equivalent to the PDF scheme in \cite{hsce6120, Fanny1} as stated in the following Corollary:
\begin{cor}
In the proposed scheme, each user transmits and decodes only one message part in the first two phases. But because of message splitting, this scheme is equivalent to a partial decode-forward scheme where each user alternatively transmits both message parts and decodes only one part in the first two phases.
\end{cor}
\vspace{-2 mm}
\begin{proof}
See Appendix A.
\end{proof}
\vspace{-3mm}
\subsection{Rate Optimization Problems}
We now turn to the question of finding the optimal resource allocation to maximize either the individual rate or the sum rate obtained by the proposed scheme as in (\ref{th1Grrsimp}). While it may be tempting to treat the individual rate optimization problem as a special case of the sum rate problem by setting $\alpha_2=0$ and $R_2=0$, the two problems are in fact distinct. The rate region for a fixed set of parameters is usually a pentagon obtained from the intersection of the two individual rate lines and a sum rate line. The sum rate line usually  cuts the axes at points outside the rate region since the sum rate constraint is larger than either individual rate constraint. Thus, the individual rate optimization problem cannot be obtained from the sum rate problem by simply setting a user's rate to zero. More specifically,
the individual rate optimization problem has only $2$ rate constraints whereas the sum rate problem has $4$ rate constraints. By setting $\alpha_2 = 0$ and $R_2 = 0$, the $4$ constraints in the sum rate problem do not reduce to the $2$ constraints in the individual rate problem.

The optimal power allocations are also different in these $2$ problems. Specifically,
\begin{itemize}
    \item For individual rate maximization, the power constraint in (\ref{powcsch2}) at user $2$ becomes $\rho_{21}=\rho_{20}=0$ and $\rho_{23}=(1-\alpha_1)^{-1}P_2$.
    \item For sum rate maximization, even with $\alpha_2=0$, the power constraint at user $2$ becomes $\rho_{21}=0$ and $\rho_{20}+\rho_{23}=(1-\alpha_1)^{-1}P_2$.
\end{itemize}
 The $\rho_{20}$ that maximizes the sum rate is not necessary zero as in the individual rate problem, thus the two problems are different.
  Consequently, we consider them separately.
\section{Individual Rate Maximization}\label{sec:indivr}
In this section, we fix the phase durations and derive the optimal power allocation that maximizes the individual rate.
We formulate the problem as a convex optimization and analyze the KKT conditions. We then obtain a theorem for the optimal schemes resulting from this optimization and
discuss its implication on practical systems such as cellular networks.
\vspace{-3mm}
 \subsection{Problem Setup}
The individual rate can be maximized for a  user when the other user has no information to send but only helps relay information of this user. In this case, the MAC-TC resembles the half-duplex relay channel in \cite{HoMa, HoMa2}.
Specifically, consider the optimal parameters that maximize $R_1$ in (\ref{th1Grrsimp}) given that $R_2=0$. We can directly see that for user $2$,
the optimal parameters are
\vspace{-1 mm}
\begin{align}\label{powal}
\alpha_2^{\star}=0,\;\rho_{22}^{\star}=\rho_{20}^{\star}=0,\; \text{and}\; \rho_{23}^{\star}=P_2/(1-\alpha_1).
\end{align}
\noindent (Throughout this paper, superscript $^\star$ indicates the optimal value.) These parameters make $J_1=S_1,$ $S_2=S_4,$ and $S_3>R_1$
in (\ref{th1Grrsimp}).

Therefore, for fixed $\alpha_1$, $R_1$ is maximized by considering the minimum between $J_1$ and $S_4$ in an optimization problem which can
be expressed as
\noindent
\begin{subequations}\label{Eq:bed5}
\begin{align}
\max_{\substack{\underline{\rho}, R_1}}\;R_1&\nonumber\\
\text{s.t.}\;R_1&\leq \alpha_1{\cal C}\!\left(g_{12}^2\rho_{11}\right)+(1-\alpha_1){\cal C}\!\left(g_{10}^2\rho_{10}\right)\label{Eq1:bed5}\\
\;R_1&\leq \alpha_1{\cal C}\!\left(g_{10}^2\rho_{11}\right)+(1-\alpha_1){\cal C}(\zeta)\label{Eq2:bed5}\\
\;P&=\alpha_1\rho_{11}+(1-\alpha_1)\left(\rho_{10}+\rho_{13}\right)\label{Eq3:bed5}\\
\;\rho_i&\geq 0,\;\text{for}\;i\in \{10,\;11,\;13\},\label{Eq4:bed5}
\end{align}
\end{subequations}
where $\zeta$ is as in (\ref{jowz}) with power allocations in (\ref{powal}), $\underline{\rho}=[\rho_{10}\; \rho_{11}\;\rho_{13}]$ and ${\cal C}(x)=\log(1+x)$. With fixed $\alpha_1$, this problem is convex and can be solved using the KKT conditions as shown subsequently.

The optimization in (\ref{Eq:bed5}) is similar to that considered in \cite{HoMa, HoMa2}. However, in these references, the formulation is for a fixed duration $(\alpha_1)$ and fixed transmit power at each phase (fixed $\rho_{11}$ and the sum $\rho_{10}+\rho_{13}$). The problem is only to find the optimal power allocations for the private $(\rho_{10})$ and the cooperative $(\rho_{13})$ message parts. Afterward, the optimal $\alpha_1$ and transmit power in each phase $(\rho_{11}$ and the  sum  $\rho_{10}+\rho_{13})$ are obtained numerically. On the one hand, this formulation simplifies the derivation of the optimal $\rho_{10}$ and $\rho_{13}$ because with fixed $\rho_{11}$ and $\rho_{10}+\rho_{13}$, $J_1$ becomes an increasing function of $\rho_{10}$ while $S_4$ becomes a decreasing function. Therefore, the optimal $\rho_{10}$ and $\rho_{13}$ can be obtained by simply solving $J_1=S_4$, if a solution exists. On the other hand, this problem formulation requires two-dimensional numerical search for the optimal $\rho_{11}$ and $\alpha_1$.

In this paper, we only fix $\alpha_1$ and analytically derive the optimal power allocations for all message parts $(\rho_{11},\rho_{10},\rho_{13})$ such that the KKT conditions for problem (\ref{Eq:bed5}) are satisfied. Then, the optimal phase duration $\alpha_1$ can be obtained numerically. Therefore, our algorithm requires numerical search for $\alpha_1$ only instead of both $\alpha_1$ and $\rho_{11}$ as in \cite{HoMa, HoMa2}.

In \cite{Fanny1}, another optimization problem is considered for minimizing the power consumption at a given transmission rate for the half-duplex relay channel.
The coding scheme in \cite{Fanny1} is slightly different from our proposed scheme in that  the source sends both message
parts in the $1^\text{st}$ phase and the relay decodes only one. This scheme
is the same as the  scheme analyzed  in Appendix A when $\alpha_2=0$. We have shown that the scheme in Appendix A and our proposed scheme
 are equivalent for the Gaussian channel since the optimal power for the non-decoded message
part at the relay is zero. Therefore, the practical designer can choose between our algorithm or that proposed in \cite{Fanny1} based on the requirement,
 whether it is to achieve the maximum transmission rate or to save power at a given data rate.
\vspace{-4 mm}
\subsection{Optimization Problem Analysis}\label{caseI}
Optimization problem (\ref{Eq:bed5}) is convex because the objective function and inequality constraints
are concave functions of $\rho_{i},\; i\in\{10,12,13\}$ while the equality constraint is linear \cite{Boyds}. Since the objective function and inequality constraints are continuously differentiable, KKT conditions are necessary
and sufficient for optimality.  The Lagrangian function for problem (\ref{Eq:bed5}) is
\begin{align}\label{lagin}
&L(R_1,\underline{\rho},\underline{\lambda},\underline{\mu})=R_1 -\lambda_0(R_1-J_1)-\lambda_1(R_1-S_4)\nonumber\\
&-\lambda_2\left(\alpha_1 \rho_{11}+(1-\alpha_1)(\rho_{10}+\rho_{13})-P\right)-\underline{\mu}\cdot\underline{\rho},
\end{align}
where $\underline{\rho}=[\rho_{10}\; \rho_{11}\; \rho_{13}]$ is the power allocation vector, $\underline{\lambda}=[\lambda_0\; \lambda_1\; \lambda_2]$ is the Lagrangian multipliers vector associated with constraints (\ref{Eq1:bed5}), (\ref{Eq2:bed5}) and (\ref{Eq3:bed5}),  $\underline{\mu}=[\mu_{10}\;\mu_{12}\;\mu_{13}]$ is the Lagrangian multipliers vector associated with the non-negative power constraints (\ref{Eq4:bed5}).
The KKT conditions are
\begin{subequations}
\begin{align}
\nabla_{\underline{\rho}}L(R_1,\underline{\rho},\underline{\lambda},\underline{\mu})&=0,\label{KKTind1}\\
P_1-\alpha_1\rho_{11}+(1-\alpha_1)\left(\rho_{10}+\rho_{13}\right)&=0,\label{KKTind2}\\
R_1-J_1&\leq 0,\label{KKTind3}\\
R_1-S_4&\leq 0,\label{KKTind4}\\
\rho_i&\geq 0,\label{KKTind5}\\
\lambda_j\geq0,\;\mu_i&\geq0,\label{KKTind6}\\
\mu_i \rho_i&=0,\label{KKTind7}\\
\lambda_0(R_1-J_1)&=0,\label{KKTind8}\\
\lambda_1(R_1-S_4)&=0.\label{KKTind9}
\end{align}
\end{subequations}
where $i\in \{10,\;12,\;13\}$, $j\in \{0,\;1,\;2\}$ and $\nabla_x f$ is the gradient of $f(.)$ at $x$. Specific expressions for $\nabla_{\underline{\rho}}L(R_1,\underline{\rho},\underline{\lambda},\underline{\mu})$ are given in Appendix B.

Depending on the relation between $g_{12}$ and $g_{10}$, several cases for problem (\ref{Eq:bed5}) can occur:
\vspace{2mm}
\newsavebox\unistress
\begin{lrbox}{\unistress}
  \begin{minipage}{0.9\textwidth}
  \vspace{-3 mm}
        \begin{align}\label{def1}
        a_1&=g_{10}^{-2}-g_{12}^{-2},\;\;\;\;\;a_2=g_{10}^{-2}+\rho_{11},\;\;\;\;\;a_3=g_{12}^{-2}+\rho_{11},\\
        a_4&=1+g_{20}g_{10}^{-1}\sqrt{\rho_{23}/\rho_{13}},\;\;\;\;\;\;\;\;\;\;\;\;\;\;\;\;\;\;
        a_5=\rho_{13}+\frac{g_{20}^2}{g_{10}^2}\rho_{23}+2\frac{g_{20}}{g_{10}}\sqrt{\rho_{13}\rho_{23}},\nonumber\\
        b_1&=a_4(P_1+g_{10}^{-2}-(1-\alpha_1)\rho_{13})+\alpha_1a_1+(1-\alpha_1)(a_4a_1+a_5),\nonumber\\
        b_2&=a_1P_1+(1-\alpha_1)\frac{g_{20}^2}{g_{10}^2}\rho_{23}+\frac{1}{g_{10}^2}a_1
        +2(1-\alpha_1)a_1\frac{g_{20}}{g_{10}}\sqrt{\rho_{13}\rho_{23}},\nonumber\\
        f_1&=\;\alpha_1 {\cal C}\!\left(g_{10}^2\rho_{11}\right)+
        (1-\alpha_1){\cal C}\!\left(g_{10}^2(\rho_{10}+a_5)\right)
        \;-\alpha_1{\cal C}\!\left(g_{12}^2\rho_{11}\right),\nonumber\\
        f_2(\rho_{13})&=\;(1-\alpha)^{-1}(P_1-\alpha \rho_{11}^{\star})-\rho_{13}-g_{10}^{-2}\left(2^{f_1/(1-\alpha)}-1\right)\nonumber\\
        f_3&=\frac{1-\alpha_1}{\alpha_1}\cdot {\cal C}\!\left(g_{10}^2a_5\right).\nonumber\\
        f_4(\rho_{13})&=\frac{P_1-(1-\alpha_1) \rho_{13}}{\alpha_1}-\frac{1}{g_{10}^2}\left(\frac{1-g_{12}^2g_{10}^{-2}}{2^{f_3}-g_{12}^2g_{10}^{-2}}-1\right),\nonumber
        \end{align}
  \vspace{1 mm}
  \end{minipage}
\end{lrbox}
\noindent
\begin{figure*}[ht]
\normalsize
\begin{center}
{\small TABLE I: ALGORITHM FOR MAXIMIZING THE INDIVIDUAL RATE.}\\
\vspace*{2mm}
\begin{tabular}{|c|c|}
\hline
\multicolumn{2}{|c|}{Definitions}\\
\hline
\multicolumn{2}{|c|}{\usebox{\unistress}}\\
\hline
Case  & Algorithm\\
\hline
$2.a:\lambda_0>0,\lambda_1>0$& $\!\!\!\!\!\!\!\!\!\!\!\!\!\!\!\!\!\!\rho_{13}^{\star}$ is obtained from solving $f_2(\rho_{13})=0$ with\\
$\;\;\;\;$and $\mu_{10}=0$ &   $
  \!\!\!\!\!\!\!\!\!\!\!\!\!\!\!\!\!\!\!
    \rho_{11}^{\star}=(2a_4)^{-1}(b_1+\sqrt{b_1^2-4a_4b_2})-g_{10}^{-2}$
  and\\
& $
\!\!\!\!\!\!\!\!\!\!\!\!\!\!\!\!\!\!\!\!\!\!\!\!\!\!\!\!\!\!\!\!\!\!\!\!\!\!\!\!\!\!\!\!
  \rho_{10}^{\star}=\left(g_{20}g_{10}^{-1}\sqrt{\rho_{23}^{\star}/\rho_{13}^{\star}}+a_3/a_2\right)$\\
  &$
  \;\;\;\;\;\;\;\;\;\;\;\;\;\;\times\left(\left(1+g_{20}g_{10}^{-1}\sqrt{\rho_{23}^{\star}/\rho_{13}^{\star}}\right)a_3-(a_3/a_2)a_5\right)
  -g_{10}^{-2}$\\
\hline
$2.b:\lambda_0>0,\lambda_1>0$ & $\!\!\!\!\!\!\!\!\!\!\!\!\!\!\!\!\!\!\!\!\!\!\!\!\!\!\!\!\!\!\!\!\!\!\!\!\!\!\!\!\!\!\!\!\!\!\!\!\!
\!\!\!\!\!\!\!\!\!\!\!\!\!\!\!\!\!\!\!\!\!\!\!\!\!\!\!\!\!\!\!\!\!\!\!\!\!\!\!\!\!\!\!\!\!\!\!\!\!\!\!\!\rho_{10}^{\star}=0$ and \\
$\;\;\;$and $\mu_{10}>0$ & $\!\!\!\!\!\!\!\!\!\!\!\!\!\!\!\!\!\rho_{13}^{\star}$ is obtained from solving $f_4(\rho_{13})=0$ with\\
 & $\!\!\!\!\!\!\!\!\!\!\!\!\!\!\!\!\!\!\!\!\!\!\!\!\!\!\!\!\!\!\!\!\!
 \!\!\!\!\!\!\!\!\!\!\!\!\!\!\!\!\!\!\!\!\!\!\rho_{11}^{\star}=\alpha_1^{-1}(P_1-(1-\alpha_1)\rho_{13}^{\star}).$\\
\hline
$3.a:\lambda_1=0, \lambda_0>0,$ & $\!\!\!\!\!\!\!\!\!\!\!\!\!\!\!\!\!\!\!\!\!\!\!\!\!\!\!\!\!\!\!\!\!\!\!\!\!\!\!\!\!\!\!\!\!\!\!
\rho_{13}^{\star}=0,\;\;\rho_{11}^{\star}=P_1+(1-\alpha_1)a_1,$\\
$\;\;\;\;\mu_{10}=0, \mu_{13}>0$ &$\!\!\!\!\!\!\!\!\!\!\!\!\!\!\!\!\!\!\!\!\!\!\!\!\!\!\!\!
\!\!\!\!\!\!\!\!\!\!\!\!\!\!\!\!\!\!\!\!\!\!\!\!\!\!\!\!\!\!\!\!\!\!\!\!\!\!\!\!\!\!\!\!\!\!\!\!\!\!\!\!\!\!\!\!\!\!\!
\rho_{10}^{\star}=P_1-\alpha_1a_1.$  \\
\hline
$3.b:\lambda_1=0,\lambda_0>0$ & $\!\!\!\!\!\!\!\!\!\!\!\!\!\!\!\!\!\!\!\!\!\!\!\!\!\!\!\!\!\!
\!\!\!\!\!\!\!\!\!\!\!\!\!\!\!\!\!\!\!\!\!\!\!\!\!\!\!\!\!\!\!\!\!\!\!\!\!\!\!\!\!\!\!\!\!\!\!\!\!\!\!\!\!\!\!\!\!\!
\!\!\!\!\!\!\rho_{13}^{\star}=\rho_{10}^{\star}=0,$\\
$\;\;\;\;\mu_{10}\!>0,\mu_{13}>0$ & $\!\!\!\!\!\!\!\!\!\!\!\!\!\!\!\!\!\!\!\!\!\!\!\!\!\!\!\!\!\!\!
\!\!\!\!\!\!\!\!\!\!\!\!\!\!\!
\!\!\!\!\!\!\!\!\!\!\!\!\!\!\!\!\!\!\!\!\!\!\!\!\!\!\!\!\!\!\!\!\!\!\!\!\!\!\!\!\!\!\!\!\!\!\!\!\!\!\!\rho_{11}^{\star}=P_1/\alpha_1.$\\
\hline
\end{tabular}
\end{center}
\vspace*{-6mm}
\end{figure*}
\begin{enumerate}
  \item Case $1$: $g_{12}< g_{10}$.\\  
  $\;\;\;$In this case, $J_1<S_4$. Therefore, (\ref{KKTind3}) can be satisfied with equality $(R_1^{\star}=J_1)$  while (\ref{KKTind4}) is satisfied without equality $(R_1^{\star}<S_4)$. Then, (\ref{KKTind4}) is an inactive constraint and from (\ref{KKTind6}), (\ref{KKTind8}) and (\ref{KKTind9}), $\lambda_0^{\star}>0$ and $\lambda_1^{\star}=0$ are optimal.
  Moreover, it is known \cite{hsce6120} that PDF relaying is helpful only when the source-relay link quality is better than the source-destination link quality. Since here $g_{12}< g_{10}$, requiring user $2$ to decode will decrease the transmission rate because it can decode at a lower rate than the destination. Therefore, user $1$ will choose to send
  its message directly to the destination $(\alpha_1^{\star}=0)$. Then, with $\alpha_1^{\star}=0$ and from the KKT conditions (\ref{KKTind2}) and (\ref{KKTind5}), the optimal parameters are
  \begin{align*}
  \alpha_1^{\star}=0,\;\rho_{11}^{\star}=\rho_{13}^{\star}=0,\;\rho_{10}^{\star}=P_1,
  \end{align*}
  which leads to the maximum rate as $R_1^{\star}={\cal C}\!\left(g_{10}^2P_1\right).$
  \item Case $2$: $g_{12}\geq g_{10}$ and there is intersection between $J_1$ and $S_4$.\\
  $\;\;\;$In this case, both inequality constraints (\ref{Eq1:bed5}) and (\ref{Eq2:bed5}) are tight. Both constraints (\ref{KKTind3}) and (\ref{KKTind4}) can be satisfied with equality $(R_1^{\star}=J_2=S_4)$
  and are active. KKT conditions
  (\ref{KKTind3}), (\ref{KKTind4}), (\ref{KKTind6}), (\ref{KKTind8}) and (\ref{KKTind9}) then lead to $\lambda_0^{\star}>0$ and $\lambda_1^{\star}>0$.
  \item Case $3$: $g_{12}\geq g_{10}$ and there is no intersection between $J_1$ and $S_4$. \\
  $\;\;\;$Here, we must have either $J_1<S_4$ or $J_1>S_4$ for all valid power allocations. However, the case $J_1>S_4$
  cannot occur because with $\rho_{13}=P_1/(1-\alpha)$, then $J_1=0$ and $S_4>0$. Therefore, $J_1<S_4$ is the only valid case. Hence, as in case $1$,
  $\lambda_0^{\star}>0$ and $\lambda_1^{\star}=0$ are optimal on KKT conditions (\ref{KKTind3}), (\ref{KKTind4}), (\ref{KKTind6}), (\ref{KKTind8}) and (\ref{KKTind9}).
  \end{enumerate}
\indent In all cases, problem (\ref{Eq:bed5}) can be solved analytically such that the KKT conditions (\ref{KKTind1}) -- (\ref{KKTind9})
are satisfied. In each case, depending on the optimal values of $\mu$ being zero or strictly positive, which means that the corresponding
power is strictly positive or zero, there can be sub-cases. Algorithms for solving these cases are given in Table I. 
See Appendix B for the proof.
\subsection{Optimal Schemes}\label{sec2c}
Based on the above algorithm, we can identify the specific scheme  optimal for each channel configuration as summarized in the following theorem:
\begin{thm}
The optimal transmission scheme for maximizing the individual rate of a half-duplex relay channel is obtained by reducing the proposed $3$-phase scheme into $2$ phases by setting $\alpha_2=0$. Then, the optimal scheme is such that in
\begin{itemize}
  \item
  \textbf{Case $1$:} User $1$ performs direct transmission to send $w_1$
  during the whole transmission time and user $2$ stays silent.
  \item
  \textbf{Case $2.a$:} For a fixed $\alpha_1$, the two users perform partial decode-forward (PDF) relaying \textbf{with} message repetition where in
  \begin{itemize}
    \item The $1^{\text{st}}$ phase: user $1$ sends $w_{12}$, user $2$ decodes it.
    \item The $3^{\text{rd}}$ phase: user $1$ sends $(w_{10},w_{12})$ and  user $2$ sends $w_{12}$.
  \end{itemize}
  \item
  \textbf{Case $2.b$}: For a fixed $\alpha_1$, the two users perform decode-forward (DF) relaying  where in
  \begin{itemize}
    \item The $1^{\text{st}}$ phase: user $1$ sends $w_1$, user $2$ decodes it.
    \item The $3^{\text{rd}}$ phase: both users send $w_1$.
  \end{itemize}
  \item
  \textbf{Case $3.a$}: For a fixed $\alpha_1$, the two users perform partial decode-forward (PDF) relaying \textbf{without} message repetition where in
  \begin{itemize}
    \item The $1^{\text{st}}$ phase: user $1$ sends $w_{12}$, user $2$ decodes it.
    \item The $3^{\text{rd}}$ phase: user $1$ sends $w_{10}$, user $2$ sends $w_{12}$.
  \end{itemize}
  \item
  \textbf{Case $3.b$:} For a fixed $\alpha_1$, the two users perform $2$-hop transmission where in
  \begin{itemize}
    \item The $1^{\text{st}}$ phase: user $1$ sends $w_{1}$, user $2$ decodes it.
    \item The $3^{\text{rd}}$ phase: user $2$ sends $w_{1}$, user $1$ stays silent.
  \end{itemize}
\end{itemize}
In all cases, the destination only decodes at the end of the $3^{\text{rd}}$
 phase based on the signals received in both phases $1$ and $3$.
\end{thm}
\begin{proof}
These optimal schemes are obtained from the cases defined in Section \ref{caseI} and the algorithm in Table I.
\end{proof}
Theorem $2$ provides a comprehensive coverage for all possible subschemes of the general PDF scheme for the half-duplex relay channel. Moreover, it provides the specific signaling for
each subscheme. Drawing on Theorem $2$, we can deduce the followings: First, direct transmission between user $1$ and the destination is preferred if the link connecting them
is stronger than the inter-user link. Second, PDF relaying with or without message repetition is optimal if the inter-user link is slightly stronger
than the link from user $1$ to the destination. Third, DF relaying is optimal if the inter-user link is much stronger than the link from user $1$ to the destination. Last,  two-hop transmission is
optimal when the inter-user link is slightly stronger than the link from user $2$ to the destination and is significantly (doubly) stronger than the link from user $1$ to the destination.

Building on
these subschemes, the practical designer of cellular, sensor or ad hoc networks can determine which subscheme to use based on the channel configuration
between the users and the destination. For example, in the uplink of a cellular network in an urban area, multiple mobiles can concurrently be active in each cell.
Two mobiles close to each other, especially with a line of sight, are likely to have strong inter-mobile link and can mutually benefit by cooperating with each other
to send their messages to the base station. Because of the strong link quality between these two mobiles, Theorem $2$ suggests that the designer can use
full decode-forward as a simple and optimal scheme for this case.
More analysis and numerical examples about the existence of each case are provided in Section \ref{sec:numresult}.
\section{Sum Rate Maximization}\label{sec:sumr}
In this section, we derive the optimal power allocation that maximizes the sum rate for fixed phase durations $(\alpha_1,\alpha_2)$. Following a procedure similar to that for the individual rate,
we first formulate the problem as a convex optimization and analyze its  KKT conditions. We then obtain a theorem for the optimal scheme for each channel
configuration.
We also consider the optimization for the symmetric channel as a special case. Last, we analyze the maximum gain for the individual and the sum rate with respect
to the classical MAC.
\subsection{Problem Setup}
 When both users have information to send, the throughput (sum rate) is an important criteria for network operators.
The throughput shows a rate limit the network can tolerate before starting losing packets or having congestion.
Define the sum rate as $S_R=R_1+R_2$. With the $4$ constraints on the sum rate in (\ref{th1Grrsimp}), the optimization problem with fixed $\alpha_1$ and $\alpha_2$ can be posed as
\noindent
\begin{subequations}\label{dersum0}
\begin{align}
\max_{\substack{\underline{\rho},S_R}}\;&S_R,\\
\text{s.t.}\;&S_R\leq S_1,\;\;S_R\leq S_2,\label{1dersum0}\\
\;&S_R\leq S_3,\;\;S_R\leq S_4,\label{2dersum0}\\
\;&P_1=\alpha_1\rho_{11}+\alpha_3\left(\rho_{10}+\rho_{13}\right),\label{3dersum0}\\
\;&P_2=\alpha_2\rho_{22}+\alpha_3\left(\rho_{20}+\rho_{23}\right),\label{4dersum0}\\
\;&\rho_i\geq 0,\;\text{for}\;i\in \{10,\;20,\;11,\;22,\;13,\;23\}\label{5dersum0}
\end{align}
\end{subequations}
\noindent 
where $\underline{\rho}=[\rho_{10}\; \rho_{20}\; \rho_{11}\; \rho_{22}\; \rho_{13}\; \rho_{23}]$ is the power allocation vector. This sum rate problem has not been considered in the literature before.

Note that $J_1+J_2$ in (\ref{th1Grrsimp}) is another constraint on the sum rate, but this constraint is redundant
since $J_1+J_2>S_1$.
The sum of the first parts in $J_1$ and $J_2$ is equal to the first two parts in $S_1$.
 However, the sum of the second parts in $J_1$ and $J_2$ is bigger than the third part in $S_1$.
 This second part in $J_1(J_2)$
 results from decoding $w_{10}(w_{20})$ at the destination without interference from other message parts.  The third part in $S_1$ results
 from decoding both $w_{10}$ and $w_{20}$ at the destination without interference from other message parts and this rate is smaller than the rates obtained
 from decoding each part alone without interference.
\subsection{Optimization Problem Analysis}\label{caseS}
Similar to individual rate optimization, problem (\ref{dersum0}) is convex because the objective function and inequality constraints
are concave functions of $\underline{\rho}$ while the equality constraints are affine \cite{Boyds}. Since the objective function and inequality constraints are continuously differentiable, KKT conditions are necessary and sufficient for optimality. The Lagrangian function for  (\ref{dersum0}) is
\begin{align}\label{lagsum}
&L(S_R,\underline{\rho},\underline{\lambda},\underline{\mu})=
S_R -\sum_{k=0}^3\lambda_k(S_R-S_{k+1})\\
&\;-
\lambda_4(\alpha_1\rho_{11}+(1-\alpha_1-\alpha_2)\left(\rho_{10}+\rho_{13}\right)-P_1)\nonumber\\
&\;-\lambda_5(\alpha_2\rho_{22}+(1-\alpha_1-\alpha_2)\left(\rho_{20}+\rho_{23}\right)-P_2)-\underline{\mu}\cdot\underline{\rho},\nonumber
\end{align}
 where $\underline{\lambda}=[\lambda_0\; \lambda_1\; \lambda_2\; \lambda_3\; \lambda_4\; \lambda_5]$ is the Lagrangian multipliers vector associated with the rate and power constraints (\ref{1dersum0}), (\ref{2dersum0}), (\ref{3dersum0}),
and (\ref{4dersum0});  $\underline{\mu}=[\mu_{10}\; \mu_{20}\; \mu_{12}\; \mu_{21}\; \mu_{13}\; \mu_{23}]$ is the Lagrangian multipliers vector associated with the non-negative power constraints in (\ref{5dersum0}).
The KKT conditions are
\begin{subequations}
\begin{align}
\nabla_{\underline{\rho}}L(S_R,\underline{\rho},\underline{\lambda},\underline{\mu})&=0,\label{KKTsum1}\\
P_1-\alpha_1\rho_{11}+(1-\alpha_1-\alpha_2)\left(\rho_{10}+\rho_{13}\right)&=0,\label{KKTsum2}\\
P_2-\alpha_2\rho_{22}+(1-\alpha_1-\alpha_2)\left(\rho_{20}+\rho_{23}\right)&=0,\label{KKTsum3}\\
S_R-S_1\leq 0,\;\;\; S_R-S_2\leq 0,&\nonumber\\
S_R-S_3\leq 0,\;\;\; S_R-S_4&\leq 0,\label{KKTsum4}\\
\rho_i&\geq 0,\label{KKTsum5}\\
\lambda_j\geq 0,\;\mu_i&\geq0,\label{KKTsum6}\\
\mu_i\rho_i&=0,\label{KKTsum7}\\
\lambda_0(S_R-S_1)=0,\;\;\; \lambda_1(S_R-S_2)=0,&\nonumber\\
\lambda_2(S_R-S_3)=0,\;\;\; \lambda_3(S_R-S_4)&=0,\;\;\;\label{KKTsum8}
\end{align}
\end{subequations}
where $i\in \{10,\;20\;12,\;22,\;13,\;23\}$ and $j\in \{0,\;1\;2,\;3,\;4,\;5\}$.
The exact expression for (\ref{KKTsum1}) is given in Appendix C.

Depending on the link qualities, problem (\ref{dersum0}) can specialize to several cases:
\begin{enumerate}
  \item Case $1$:  $g_{12}\leq g_{10}$ and $g_{21}\leq g_{20}$.\\
  In this case, we can see from the sum rate expressions in (\ref{th1Grrsimp}) that $S_1$ is the minimum among $(S_1,S_2,S_3,S_4)$. The constraint $S_R-S_1=0$ is then the only active constraint in (\ref{KKTsum4}). Therefore, KKT conditions  (\ref{KKTsum4}) and (\ref{KKTsum8}) lead to $\lambda_1^{\star}=\lambda_2^{\star}=\lambda_3^{\star}=0$ and $\lambda_0^{\star}>0$.
  Moreover, similar to the individual rate case, since $g_{12}<g_{10}$ and $g_{21}<g_{20}$, requiring each
user to decode part of the message of the other user will decrease the rate. The two users will choose to send their messages directly to the destination instead of
cooperating.  
  With $\alpha_1^{\star}=\alpha_2^{\star}=0$ and KKT conditions (\ref{KKTsum2}), (\ref{KKTsum3}) and (\ref{KKTsum5}), the optimal scheme is obtained by setting
  \begin{align}\label{moe1}
  \alpha_1^{\star}&=\alpha_2^{\star}=0,\;\;\rho_{11}^{\star}=\rho_{22}^{\star}=\rho_{13}^{\star}=\rho_{23}^{\star}=0,\nonumber\\
  \rho_{10}^{\star}&=P_1,\; \rho_{23}^{\star}=P_2,
  \end{align}
  \noindent which resembles the classical MAC with the maximum sum rate as $R_s^{\max}=C(g_{10}^2P_1+g_{20}^2P_2).$
  \item Case $2$:  $g_{12}>g_{10}$ and $g_{21}>g_{20}$.\\
In this case, $S_2$ and $S_3$ are redundant because both are bigger than $S_4$ as noticed from (\ref{th1Grrsimp}). The constraints $S_R-S_2\leq0$ and $S_R-S_3\leq0$ in (\ref{KKTsum4}) are inactive and KKT conditions (\ref{KKTsum4}) and (\ref{KKTsum8}) lead to $\lambda_1^{\star}=\lambda_2^{\star}=0$. Moreover, we can show that neither $S_1 < S_4$ nor $S_1 > S_4$ holds for all power allocations that satisfy the power constraints. $S_4$ is not always less than $S_1$ because with $\rho_{10}=\alpha_3^{-1}P_1$ and $\rho_{20}=\alpha_3^{-1}P_2$, then $S_4=S_1$. Also, $S_1$ is not always less than $S_4$ because when maximizing $S_1$ alone, we can directly notice that $\rho_{13}^{\star}=\rho_{23}^{\star}=0$; then, regardless of the values of $\rho_{11}^{\star},$ $\rho_{22}^{\star},$ $\rho_{10}^{\star},$ and
$\rho_{20}^{\star}$, we obtain $S_4<S_1$ since $g_{12}>g_{10}$ and $g_{21}>g_{20}$. Therefore, both constraints $(S_R-S_1\leq0$ and $S_R-S_4\leq0)$ must be tight and active,  and KKT conditions  (\ref{KKTsum4})
and (\ref{KKTsum8}) lead to $\lambda_0^{\star}>0$ and $\lambda_3^{\star}>0.$
  \item Case $3$: $g_{12}>g_{10}$ and $g_{21}\leq g_{20}$.\\
Since $g_{12}>g_{10}$, user $1$ will send part of its information via user $2$. However, because $g_{21}\leq g_{20}$, user $2$ will send its information directly to the destination. From KKT condition (\ref{KKTsum3}), the optimal scheme is obtained by setting $\alpha_2^{\star}=0$ and $\rho_{22}^{\star}=0$. By substituting these values into (\ref{th1Grrsimp}), we obtain $S_2=S_4$ and $S_3>S_4$. KKT conditions  (\ref{KKTsum4})
and (\ref{KKTsum8}) lead to  $\lambda_1^{\star}=\lambda_2^{\star}=0$. Therefore, for a fixed $\alpha_1$, we have the same optimization problem as in Case $2$ but with $\alpha_2^{\star}=0$ and $\rho_{22}^{\star}=0$.
  \item Case $4$: $g_{12}\leq g_{10}$ and $g_{21}>g_{20}$\\
This case is the opposite of Case $3$. For a fixed $\alpha_2$, problem (\ref{dersum0}) is considered with $\alpha_1^{\star}=0$, and $\rho_{11}^{\star}=0$.
\end{enumerate}
For a fixed pair $(\alpha_1, \alpha_2)$, solutions for cases $2,$ $3,$ $4$ can be obtained analytically such that KKT conditions in (\ref{KKTsum1})--(\ref{KKTsum8})
are satisfied as shown in Appendix C. The solutions are given in Table II with the following definitions:
\begin{align}\label{def2}
        \!\!a_1&=g_{10}^{-2}-g_{12}^{-2},\;\;\;\;\;\;a_2=g_{20}^{-2}-g_{21}^{-2},\\
        \!\!a_3&=g_{10}^{-2}+
        \rho_{11},\;\;\;\;\;\;\;\!a_4=g_{20}^{-2}+\rho_{22},\nonumber\\
        \!\!a_5&=\rho_{11}+g_{12}^{-2},\;\;\;\;\;\;\;\!a_6=1+g_{10}^2\rho_{10}+g_{20}^2\rho_{20},\nonumber\\
        \!\!a_7&=2g_{20}^2,\;\;\;\;\;\;\;\;\;\;\;\;\;\;\;a_8=\frac{2}{1-\alpha_1}g_{10}^2,\;\;\nonumber\\
        \!\!a_9&=1+g_{10}^2\rho_{10}+g_{20}^2(1-\alpha_1)^{-1}P_2,\nonumber\\
        \!\!a_{10}&=(1-\alpha_1)^{-1}(g_{10}^2P_1+\alpha_1)-g_{10}^2\rho_{10},\nonumber\\
        \!\!a_{11}&=\frac{1+g_{10}^2P_1+g_{20}^2P_2}{1-\alpha_1},\;\;\;\;\;\;\nonumber\\
        \!\!b_1&=2g_{20}^2(a_2+a_3)+a_6(2-b_3),\nonumber\\
        \!\!b_2&=\frac{2+\alpha_1}{1-\alpha_1}g_{10}^2a_1+2a_{11},\;\;\;b_3=\frac{2a_3-a_1}{a_3-a_1}\nonumber\\
        \!\!b_4&=a_1(a_9+3a_{10}),\;\;\;\;\;\;\;\;\;b_5=a_2(a_6+2g_{10}^2a_3-\mu a_6),\nonumber\\
        \!\!b_6&=c_2g_{20}^2+g_{10}g_{20}\left(2\sqrt{\rho_{13}}-\frac{a_5+a_1}{\sqrt{\rho_{13}}}\right),\nonumber\\
        \!\!b_7&=a_{12}c_5,\;\;\;\;\;
        \;\;\;\;\;\;\;\;\;\;\;b_8=b_2+g_{10}^2\rho_{13}-g_{10}^2(a_5+a_1),\nonumber\\
        \!\!b_9&=g_{10}g_{20}^{-1}(a_5+a_1)a_5^{-1}\sqrt{\rho_{13}}\nonumber\\
        \!\!b_{10}&=\frac{(a_5+a_1)g_{10}\!\!\left(g_{10}+g_{20}\sqrt{\frac{\rho_{23}}{\rho_{13}}}\right)\!\!-\!\!
        \left(g_{10}\sqrt{\rho_{13}}+g_{20}\sqrt{\rho_{23}}\right)^2}
        {1+(a_5+a_1)a_5^{-1}\frac{g_{10}}{g_{20}}\sqrt{\frac{\rho_{13}}{\rho_{23}}}}.\nonumber\\
        \!\!f_1&=\alpha_3^{-1}\big[\alpha_1{\cal C}\!\left(g_{12}^2\rho_{11}\right)+\alpha_2{\cal C}\!\left(g_{21}^2\rho_{22}\right)\nonumber\\
        &\;\;\;
        -\alpha_1{\cal C}\!\left(g_{10}^2\rho_{11}\right)-\alpha_2{\cal C}\!\left(g_{20}^2\rho_{22}\right)\big],\nonumber\\
        \!\!f_2(a_6)&=2g_{10}^2a_3 -(a_3-a_1)^{-1}(2a_3-a_1)a_6-(2^{f_1}-1)a_6\nonumber\\
        \!\!f_3(\rho_{11})&=0.5(2g_{10}^2a_3 -\frac{2a_3-a_1}{a_3-a_1}a_6)+a_6-1\nonumber\\
        &\;\;-\frac{g_{10}^2P_1+g_{20}^2P_2-\alpha_1g_{10}^2\rho_{11}-\alpha_2g_{20}^2\rho_{22}}{1-\alpha_1-\alpha_2},\nonumber\\
        \!\!f_4(\rho_{23})&=\left(1+\left(g_{10}\sqrt{\rho_{13}}+g_{20}\sqrt{\rho_{23}}\right)^2\right)\cdot
        \left((a_1/a_3)-(a_2/a_4)\right)\nonumber\\
        &\;\;+g_{10}\left(g_{10}+g_{20}\sqrt{\rho_{23}/\rho_{13}}\right)(a_3-a_1)\nonumber\\
        &\;\;-g_{20}\left(g_{20}+g_{10}\sqrt{\rho_{13}/\rho_{23}}\right)(a_4-a_2),\nonumber\\
        \!\!f_5(\rho_{13})&=\rho_{13}-g_{10}^{-2}\left(\sqrt{2^{f_1}-1}-g_{20}\sqrt{\rho_{23}}\right)^2\nonumber\\
        \!\!f_6&=(\alpha_1/\alpha_3)\left({\cal C}\!\!\left(g_{12}^2\rho_{11}\right)-{\cal C}\!\!\left(g_{10}^2\rho_{11}\right)\right),\nonumber\\
        \!\!f_7(\rho_{10})&=2g_{10}^2a_3-\frac{2a_3-a_1}{a_3-a_1}(1+g_{10}^2\rho_{10}+g_{20}^2\rho_{20})\nonumber\\
        &\;\;-(1+g_{10}^2\rho_{10}+g_{20}^2\rho_{20})(2^{f_6}-1),\nonumber\\
        \!\!f_8(\rho_{13})&=b_{10}-(2^{f_1}-1)^{-1}\left(g_{10}\sqrt{\rho_{13}}+g_{20}\sqrt{\rho_{23}}\right)^2,\nonumber
        \end{align}
\noindent
\begin{figure*}[ht]
\normalsize
\begin{center}
{\small TABLE II: ALGORITHM FOR MAXIMIZING THE SUM RATE.}\\
\vspace*{2mm}
\begin{tabular}{|c|c|}
\hline
Case  & Algorithm\\
\hline
$2.a$& $\!\!\!\!\!\!\!\!\!\!\!\!\!\!\!\!\!\!\!\!\!\!\!\!\!\!\!\!\!\!\!\!\!\!\!\!\!\!\!\!\!\!\!\!\!\!\!\!\!\!\!\!\!\!\!\!\!\!\!\!\!\!\!\!\!a_6^{\star}$ in (\ref{def2}) is obtained from solving $f_2(a_6)=0$ with \\
$\mu_{10}=0,$ & $\!\!\!\!\!\!\!\!\!\!\!\!\!\!\!\!\!\!\!\!\!\!\!\!\!\!\!\!\!\!\!\!\!\!\!\!\!\!\!\!\!\!\!\!\!\!\!\!\!\!\!\!\!\!\!\!\!\!\!\!\!\!\!\!\!\!\!\!\!\!\!\!\!\!\!\!
\!\!\!\!\!\!\!\!\!\!\!\!\!\!\!\!\!\!\!\!\!\!\rho_{11}^{\star}$ obtained from solving $f_3(\rho_{11})=0$, \\
$\mu_{20}=0$ &   $\rho_{22}^{\star}=(2a_7)^{-1}(b_1+\sqrt{b_1^2-4a_7b_2})-g_{20}^{-2},\;\;\;\rho_{13}^{\star}=0.5a_3-\frac{b_3 a_6^{\star}}{4g_{10}^2},\;\rho_{23}^{\star}=\frac{g_{10}^2}{g_{20}^2}\rho_{13}^{\star},$\\
 & $\!\!\!\!\!\!\!\!\!\!\!\!\rho_{10}^{\star}=\alpha_3^{-1}(P_1-\alpha_1\rho_{11}^{\star})-\rho_{13}^{\star},$ and $\;\;\;\;\;\;\;\;\;\;\;\;\rho_{20}^{\star}=\alpha_3^{-1}(P_2-\alpha_2\rho_{22}^{\star})-\rho_{23}^{\star}$\\
\hline
$2.b:$ & $\!\!\!\!\!\!\!\!\!\!\!\!\!\!\!\!\!\!\!\!\!\!\!\!\!\!\!\!\!\!\!\!\!\!\!\!\!\!\!\!\!\!\!\!\!\!\!\!\!\!\!\!\!\!\!\!\!\!\!\!\!\!\!\!\!\!\!\!\!\!
\!\!\!\!\!\!\!\!\!\!\!\!\!\!\!\!\!\!\!\!\!\!\!\!\!\!\!\!\!\!\!\!\!\!\!\!\!\!\!\!\!\!\!\!\!\!\!\!\!\!\!\!\!\!\!\!\!\!\!\!\!\!\!\!\!\!\!\!\!\!\!\!\!
\!\!\!
\rho_{10}^{\star}=\rho_{20}^{\star}=0$ and \\
$\mu_{10}>0,$ & $\!\!\!\!\!\!\!\!\!\!\!\!\!\!\!\!
\!\!\!\!\!\!\!\!\!\!\!\!\!\!\!\!\!\!\!\!\!\!\!\!\!\!\!\!\!\!\!\!\!\!\!\!\!\!\!\!\!\!\!\!\!\!\!\!\!\!\!\!\!\!\!\!\!\!\!\!\!
\rho_{13}^{\star}$ is obtained from solving $f_5(\rho_{13})=0$ with \\
$\mu_{20}>0$ & $\!\!\!\!\!\!\!\!\!\!\!\!\!\!\!\!\!\!\!\!\!\!\!\!\!\!\!\!\!\!\!\!\!\!\!\!
\!\!\!\!\!\!\!\!\!\!\!\!\!\!\!\!\!\!\!\!\!\!\!\!\!\!\!\!\!\!\!\!\!\!\!\!\!\!\!\!\!\!\!\!\!\!\!\!\!\!\!\!\!\!\!\!\!\!\!
\rho_{23}^{\star}$ obtained from solving $f_4(\rho_{23})=0,$  \\
& $\!\!\!\!\!\!\!\!\!\!\!\!\!\!\!\!\!\!\!\!\!\!\!\!\!\!\!\!\!\!\!\!\!\!\!\!\!\!\!\!\!\!\!\!\!\!\!\!\!\!\!\!\!\!\!\!\!\!
\rho_{11}^{\star}=\alpha_1^{-1}(P_1-\alpha_3\rho_{13}^{\star}),$ and  $\rho_{22}^{\star}=\alpha_2^{-1}(P_2-\alpha_3\rho_{23}^{\star}),$ \\
\hline
$3.a$& $\!\!\!\!\!\!\!\!\!\!\!\!\!\!\!\!\!\!\!\!\!\!\!\!\!\!\!\!\!\!\!
\!\!\!\!\!\!\!\!\!\!\!\!\!\!\!\!\!\!\!\!\!\!\!\!\!\!\!\!\!\!\!\!\!\!\!\!\!\!\!\!\!\!\!\!\!\!\!
\rho_{10}^{\star}$ is obtained from solving $f_7(\rho_{10})=0$ with \\
$\mu_{10}=0,$ & $\!\!\!\!\!\!\!\!\!\!\!\!\!\!\!\!\!\!\!\!\!\!\!\!\!\!\!\!\!\!\!\!\!\!\!\!\!\!\!\!\!\!\!\!\!\!\!
\!\!\!\!\!\!\!\!\!\!\!\!\!\!\!\!\!\!\!\!\!\!\!\!\!\!\!\!\!\!\!\!\!\!\!\!\!\!\!\!\!\!
\rho_{11}^{\star}=(2a_8)^{-1}(b_5+\sqrt{b_5^2-4a_8b_4})-g_{10}^{-2},$\\
$\mu_{20}=0$ &   $\!\!\!\!\!\!\!\!\!\!\!\!\!\!\!\!\!\!\!\!\!\!\!\!\!\!\!\!\!\!\!\!\!\!\!\!\!\!\!\!\!\!\!
\rho_{20}^{\star}=(1-\alpha_1)^{-1}P_2-g_{10}^2g_{20}^{-2}
\big((1-\alpha_1)^{-1}(P_1-\alpha_1)-\rho_{10}^{\star}\big),$\\
& $\!\!\!\!\!\!\!
\!\!\!\!\!\!\!\!\!\!\!\!\!\!\!\!\!\!\!\!\!\!\!\!\!\!\!\!\!\!\!\!\!\!\!\!\!\!\!\!\!\!\!
\rho_{13}^{\star}=0.5a_3-(a_3-a_1)(2a_3-a_1)(1+g_{10}^2\rho_{10}^{\star}+g_{20}^2\rho_{20}^{\star}),$\\
& $\!\!\!\!\!\!\!\!\!\!\!\!\!\!\!\!\!\!\!\!\!\!\!\!\!\!\!\!\!\!\!\!\!\!\!\!\!\!\!\!\!\!\!\!\!\!\!
\!\!\!\!\!\!\!\!\!\!\!\!\!\!\!\!\!\!\!\!\!\!\!\!\!\!\!\!\!\!\!\!\!\!\!\!\!\!\!\!\!\!\!\!\!\!\!\!\!\!\!\!\!\!\!
\!\!\!\!\!\!\!\!\!\!\!\!\!\!\!\!\!\!\!\!\!\!\!\!\!\!\!\!\!\!\!\!\!\!\!\!\!\!\!\!\!\!\!\!\!\!\!\!
\rho_{23}^{\star}=g_{10}^2g_{20}^{-2}\rho_{13}^{\star}.$\\
\hline
$3.b: $ & $\!\!\!\!\!\!\!\!\!\!\!\!\!\!\!\!\!\!\!\!\!\!\!\!\!\!\!\!\!\!\!\!\!\!\!\!\!\!\!\!\!\!\!\!\!\!\!\!\!\!\!\!\!\!\!
\!\!\!\!\!\!\!\!\!\!\!\!\!\!\!\!\!\!\!\!\!\!\!\!\!\!\!\!
\!\!\!\!\!\!\!\!\!\!\!\!\!\!\!\!\!\!\!\!\!\!\!\!\!\!\!\!\!\!\!\!\!\!\!\!\!\!\!\!\!\!\!\!\!\!\!\!\!\!\!\!\!\!\!\!\!\!\!\!\!\!\!\!\!\!\!\!\!\!\!\!\!\!\!\!\!\!\!\!
\rho_{10}^{\star}=0$ and \\
 $\mu_{10}>0,$ & $
 \!\!\!\!\!\!\!\!\!\!\!\!\!\!\!\!\!\!\!\!\!\!\!\!\!\!\!\!\!\!\!\!\!\!\!\!\!\!\!\!\!\!\!\!\!\!\!\!\!\!\!\!\!\!\!\!\!\!\!\!\!\!\!\!\!\!\!\!\!\!\!\!\!\!\!\!\!\!\!
 \rho_{13}^{\star}$ is obtained from solving $f_8(\rho_{13})=0$ with \\
$\mu_{20}=0$ & $\!\!\!\!\!\!\!\!\!\!\!\!\!\!\!\!\!\!\!\!\!\!
\!\!\!\!\!\!\!\!\!\!\!\!\!\!\!\!\!\!\!\!\!\!\!\!\!\!\!\!\!\!\!\!\!\!\!\!\!\!\!\!\!\!\!\!\!\!\!\!\!\!\!\!\!\!\!\!\!\!\!\!\!\!\!\!\!\!\!\!\!\!\!\!\!\!\!\!\!
\rho_{23}^{\star}=\left((2b_6)^{-1}(b_8+\sqrt{b_8^2-4b_6b_7})\right)^2$  \\
& $\!\!\!\!\!\!\!\!\!\!\!\!\!\!\!\!\!\!\!\!\!\!\!\!\!\!\!\!\!\!\!\!\!\!\!\!\!\!\!\!
\!\!\!\!\!\!\!\!\!\!
\!\!\!\!\!\!\!\!\!\!\!\!\!\!\!\!\!\!\!\!\!\!\!\!\!\!\!\!\!\!\!\!\!\!\!\!\!\!\!\!\!\!\!\!\!\!\!\!\!\!\!\!\!\!\!\!\!\!\!\!\!\!\!\!\!\!\!\!\!\!\!\!\!\!\!\!
\rho_{20}^{\star}=(1-\alpha_1)^{-1}P_2-\rho_{23}^{\star},$\\
&$\!\!\!\!\!\!\!\!\!\!\!\!\!\!\!\!\!\!\!\!\!\!\!\!\!\!\!\!\!\!\!\!\!\!\!\!\!\!\!\!\!\!
\!\!\!\!\!\!\!\!\!\!\!\!\!\!\!\!\!\!\!\!\!\!\!\!\!\!\!\!\!\!\!\!\!\!\!\!\!\!\!\!\!\!\!\!\!\!\!\!\!\!\!\!\!\!\!\!\!\!\!\!\!\!\!\!\!\!\!\!\!\!\!\!\!\!\!\!\!\!\!
\rho_{11}^{\star}=\alpha_1^{-1}(P-(1-\alpha_1))\rho_{13}^{\star}$\\
\hline
$4$& similar to Case $3$ but interchanging the power allocation between two users.\\
\hline
\end{tabular}
\end{center}
\vspace*{-6mm}
\end{figure*}
\vspace{-8 mm}
\subsection{Optimal Schemes}
As in the individual rate case, for each channel configuration, different optimal scheme for maximizing the sum rate
can result from the proposed
scheme. Theorem $3$ summarizes them.
\begin{thm}
The optimal scheme for maximizing the sum rate in a half-duplex MAC-TC is
\begin{itemize}
  \item
  \textbf{Case $1$:} Both users send their messages during the whole transmission time without cooperation as in the classical MAC.
     \item
  \textbf{Case $2.a$:} For a given pair $(\alpha_1,\alpha_2)$, both users perform PDF relaying where in
  \begin{itemize}
    \item The $1^{\text{st}}$ phase:  user $1$ sends $w_{12}$, user $2$ decodes it.
    \item The $2^{\text{nd}}$ phase:  user $2$ sends $w_{21}$, user $1$ decodes it.
    \item The $3^{\text{rd}}$ phase: user $1$ sends $(w_{10},w_{12},w_{21})$ and user $2$ sends $(w_{20},w_{12},w_{21})$.
  \end{itemize}
   \item
  \textbf{Case $2.b$:} For a given pair $(\alpha_1,\alpha_2)$, both users perform DF relaying where in
  \begin{itemize}
    \item The $1^{\text{st}}$ phase: user $1$ sends $w_{1}$, user $2$ decodes it.
    \item The $2^{\text{nd}}$ phase: user $2$ sends $w_{2}$, user $1$ decodes it.
    \item The $3^{\text{rd}}$ phase: each user sends both messages $(w_1,w_2)$.
  \end{itemize}
  \item
  \textbf{Case $3.a$:} For a given $\alpha_1$ (here $\alpha_2^{\star} = 0$), user $2$ performs direct transmission and user $1$ performs PDF relaying where in
  \begin{itemize}
    \item The $1^{\text{st}}$ phase: user $1$ sends $w_{12}$, user $2$ decodes it.
    \item The $3^{\text{rd}}$ phase: user $1$ sends $(w_{10},w_{12})$ and user $2$ sends $(w_2,w_{12})$.
  \end{itemize}
  \item
  \textbf{Case $3.b$:} For a given $\alpha_1$ (here $\alpha_2^{\star} = 0$), user $2$ performs direct transmission and user $1$ performs DF relaying where in
  \begin{itemize}
    \item The $1^{\text{st}}$ phase: user $1$ sends $w_{1}$, user $2$ decodes it.
    \item The $3^{\text{rd}}$ phase: user $1$ sends $w_{1}$ and user $2$ sends $(w_2,w_1)$.
  \end{itemize}
  \item \textbf{Case $4$}: This case is simply the opposite of Case $3$.
\end{itemize}
In all cases, the destination jointly decodes all the messages only at the end of the $3^{\text{rd}}$ phase using signals received in all three phases.
\end{thm}
\begin{proof}
These optimal schemes are obtained from the cases defined in Section \ref{caseS} and the algorithm in Table II.
\end{proof}
Theorem $3$ covers all possible subschemes for maximizing the sum rate of the MAC-TC with partial decode-forward relaying. The optimal transmission scheme at each user varies from direct transmission to PDF to DF relaying as the link to the other user respectively varies from weaker to slightly to significantly stronger than the link between this user and the destination. The following conclusions can be drawn from Theorem $3$. First, direct transmission from both users is preferred if each inter-user link is weaker than the link from the respective user to the destination. Second, when the inter-user links are stronger than the user-destination links,  DF or PDF scheme from both users is optimal depending on whether the inter-user links are slightly or significantly stronger than the user-destination links. Third, when user $1$ has a stronger link to user $2$ than to the destination while user $2$ has a weaker link to user $1$ than to the destination, user $2$ chooses direct transmission while user $1$ performs PDF or DF relaying depending on whether its link to user $2$ is slightly  or significantly stronger than its link to the destination. More analysis for the existence of each subscheme is given Section \ref{sec:numresult}.

\indent For practical application pertaining to the uplink in cellular networks, this theorem allows the designer to determine the best scheme for each channel configuration. Especially when the two mobiles are close to each other such that they have very strong inter-link qualities, the theorem suggests the use of simple decode-forward for maximizing the sum rate transmission from these two mobiles to the base station.\\
\indent Extension to more than two users is possible and a coding scheme for $m$-user MAC-TC is given in \cite{haivu3}. Although the power allocation problem is similar to the two-user case, it will be more complicated because of the increase in the number
of cases and power parameters to be optimized. Nevertheless, the optimal scheme at each user is expected to be similar to the two-user case in that it moves from direct
transmission to partial DF to full DF relaying as the inter-user links move from weaker to slightly to significantly stronger than the user-destination links.
\vspace{-2 mm}
\subsection{Special Case of Symmetric Channels}
Symmetric channels can occur when quantizing the channel coefficients as done in practice, where the destination sends feedback to the two users about quantized  channel coefficients. The destination employs round-down quantization to ensure that the rates generated from the quantized coefficients are achievable. The quantized coefficients will be the same if the actual coefficients are close. Furthermore, the inter-user links are the same based on reciprocity either with or without quantization. Hence, symmetric channels can occur for a non-negligible range of channel links in practice. For example, in WiFi networks, the access point and computers usually have fixed locations during use and the surrounding environment is almost stable such that distance becomes the main factor affecting link qualities. In such networks, symmetric channel occurs if two computers have approximately similar distances
to the access point because of quantized channel coefficients.

In this section, we briefly analyze the optimal power allocation for the special case of symmetric channels. The optimization problem is simpler with more closed-form results.
In \cite{haija1}, we have derived the optimal parameters for a slightly different coding scheme in which each user splits its message into three parts; however,
we show in \cite{haivu3} that it is equivalent to the scheme
 considered in this paper.

Consider a symmetric channel with $g_{10}=g_{20},$ $g_{12}=g_{21},$ and $P_1=P_2=P$.
Then,  $\rho_{10}^{\star}=\rho_{20}^{\star},$ $\rho_{11}^{\star}=\rho_{22}^{\star},$ $\rho_{13}^{\star}=\rho_{23}^{\star},$  and
 $\alpha_1^{\star}=\alpha_2^{\star}=\alpha^{\star}$. For the case $g_{10}\geq g_{12}$, the optimal parameters are that of the classical MAC (i.e.
 no cooperation). 
 For the case $g_{12}> g_{10}$, we have KKT conditions similar to those in (\ref{KKTsum1})--(\ref{KKTsum8}). The optimal parameters  can be obtained for a fixed
 $\alpha$ as in Table III. The proof is similar to that for the sum rate in Appendix C.
We can then vary $\alpha$, find the corresponding maximum sum rate for each $\alpha$ and choose the optimal $\alpha^{\star}$ that
corresponds to the maximum rate overall.
\vspace{2 mm}
\newsavebox\unistrainb
\begin{lrbox}{\unistrainb}
  \begin{minipage}{0.9\textwidth}
  \vspace{-2 mm}
        \begin{align}\label{intersy}
        a_1\;\;&\text{and}\;\; a_2\;\; \text{are the as in (\ref{def1}).}\nonumber\\
        a_3=&\;4\rho_{13},\;\;\;\;\;\;\;\;\;\;\;\;\;\;\;\;\;\;\;\;\;\;\;\;\;\;\;\;\;\;\;b_1=(2a_2-a_1)^{-1}((a_2-a_1)(2a_2-a_3)),\; \;\;\nonumber\\
        b_2=&\;(1-\alpha)a_1+g_{10}^{-2}+2P,\;\;\;\;  \mu_3=0.25a_1a_3(1-2\alpha)+0.5a_1\left(2P+g_{10}^{-2}\right).\nonumber\\
        f_1=&\;2\alpha {\cal C}\!\!\left(g_{10}^2\rho_{11}\right)-2\alpha {\cal C}\!\!\left(g_{12}^2\rho_{11}\right)
        +(1-2\alpha)\log\left(g_{10}^2(b_1+a_3)\right)\nonumber\\
        f_2(\rho_{13})=&\;(1-2\alpha)^{-1}(P-\alpha \rho_{11})-\rho_{13}-(2g_{10}^2)^{-1}\left(2^{\frac{f_1}{1-2\alpha}}-1\right),\nonumber\\
        f_3=&\;\frac{1-2\alpha}{2\alpha}\log\left(1+4g_{10}^2\rho_{13}\right).\nonumber\\
        f_4(\rho_{13})=&\;\frac{1}{\alpha}(P-(1-2\alpha)\rho_{13})-\frac{1}{2g_{12}^2}\left(2^{f_3}-1\right).
        \end{align}
        \vspace{1 mm}
  \end{minipage}
\end{lrbox}
\noindent
\begin{figure*}[ht]
\normalsize
\begin{center}
{\small TABLE III: ALGORITHM FOR MAXIMIZING THE SUM RATE FOR A SYMMETRIC CHANNEL.}\\
\vspace*{2mm}
\begin{tabular}{|c|c|}
\hline
\multicolumn{2}{|c|}{Definitions}\\
\hline
\multicolumn{2}{|c|}{\usebox{\unistrainb}}\\
\hline
Case  & Algorithm\\
\hline
$2.a:$& $\rho_{13}^{\star}$ is obtained from solving $f_2(\rho_{13})=0$ with\\
$\mu_{10}=0$ &   $\!\!\!\!\!\!\!\!\!\!\!\!\!\!\!\!\rho_{11}^{\star}=0.5\left(b_2+\sqrt{b_2^2-4b_3}\right)-g_{10}^{-2}$ and \\
& $\!\!\!\!\!\!\!\!\!\!\!\!\!\!\!\!\!\!\!\!\!\!\!\!\!\!\rho_{10}^{\star}=(1-2\alpha)^{-1}(P-\alpha \rho_{11}^{\star})-\rho_{13}^{\star}$\\
\hline
$2.b:$ & $\!\!\!\!\!\!\!\!\!\!\!\!\!\!\!\!\!\!\!\!\!\!\!\!\!\!\!\!\!\!\!\!\!\!\!\!\!\!\!\!\!\!\!\!\!\!\!\!\!\!\!\!\!\!\!\!\!\!\!\!\!\!\!\!
\!\!\!\!\!\!\!\!\!\!\!\!\!\!\!\!\!\!\!\!
\rho_{10}^{\star}=0$ and \\
$\mu_{10}>0$ &$\rho_{13}^{\star}$ is obtained from solving $f_4(\rho_{13})=0$ with \\
 & $\!\!\!\!\!\!\!\!\!\!\!\!\!\!\!\!\!\!\!\!\!\!\!\!\!\!\!\!\!\!\!\!\!\!\!\!\!\!\!\!\rho_{11}^{\star}=\alpha^{-1}(P-(1-2\alpha)\rho_{13}^{\star}).$\\
\hline
\end{tabular}
\end{center}
\vspace*{-4mm}
\end{figure*}
\vspace{-2 mm}
\subsection{Maximum Gain with Respect to the Classical MAC}
We now analyze the maximum gain that can be obtained by the proposed cooperative scheme  compared with the classical non-cooperative MAC. To find the maximum gain,
 we consider the asymptote where  $(g_{12},g_{21})\rightarrow\infty$ and $(P_1,P_2)\rightarrow\infty$ and obtain the following theorem:
\begin{thm}
The maximum gain obtained by the proposed cooperative scheme compared with the classical MAC when $(g_{12},g_{21})\rightarrow\infty$
and $(P_1,P_2)\rightarrow\infty$ are
\begin{align}\label{gain1}
\Delta(R_1)&=\;{\cal C}\!\left(\frac{g_{20}^2+2g_{10}g_{20}}{g_{10}^2}\right),\nonumber\\
\Delta(R_2)&=\;{\cal C}\!\left(\frac{g_{10}^2+2g_{10}g_{20}}{g_{20}^2}\right),\nonumber\\
\Delta(R_1+R_2)&=\;{\cal C}\!\left(\frac{2g_{10}g_{20}}{g_{10}^2+g_{20}^2}\right).
\end{align}
Furthermore, for symmetric channels with $g_{10}=g_{20}$, the maximum gain is equal to $2$ bps/Hz for the individual
rate and $1$ bps/Hz for the sum rate.
\end{thm}
\begin{proof}
In the asymptote as $(g_{12},g_{21})\rightarrow\infty$ in (\ref{th1Grrsimp}), $S_4$ becomes the only active constraint. The maximum of $S_4$ is achieved with $\alpha_3\rightarrow 1$, $\rho_1\rightarrow P_1$, and $\rho_2\rightarrow P_2$. Comparing $S_4^{\max}$ with the individual and sum rates for the classical MAC \cite{hsce6120}, we obtain the following formula:
\begin{align}\label{gain2}
\!\!\!\!S_4^{\max}-R_1=&\;{\cal C}\!\!\left(\frac{g_{20}^2P_2+2g_{10}g_{20}\sqrt{P_1P_2}}{1+g_{10}^2P_1}\right),\nonumber\\
\!\!\!\!S_4^{\max}-R_2=&\;{\cal C}\!\!\left(\frac{g_{10}^2P_1+2g_{10}g_{20}\sqrt{P_1P_2}}{1+g_{20}^2P_2}\right)\nonumber\\
\!\!\!\!S_4^{\max}\!-(R_1+R_2)\!=&\;{\cal C}\!\!\left(\frac{2g_{10}g_{20}\sqrt{P_1P_2}}{1+g_{10}^2P_1+g_{20}^2P_2}\right)\!\!.
\end{align}
\noindent Letting $(P_1,P_2)\rightarrow \infty$, we obtain (\ref{gain1}).
\end{proof}
We will see from simulation results later that these asymptotic gains are closely reached even at reasonable link gains and finite transmit powers.
\section{Optimal Phase Duration}
We have discussed algorithms for optimizing the power allocation to maximize the individual and the sum rates with fixed phase durations. These phase durations are indeed often fixed in practical applications, for example in GSM systems (European standard for cellular networks).
In GSM, the available band for either uplink or downlink communication is $25$MHz. This band is divided into $125$ sub-channels of $200$KHz each using frequency division. Each sub-channel is shared by $8$ users using time division with $526.92\mu s$ for each time slot or phase. The proposed algorithms can be applied there directly.

When the optimal phase duration is of interest, we can use a numerical search method. In this section, we discuss a simple grid search and its impact on the  running time  as the number of users increases. We also propose a simple and fast interpolation method that can achieve the optimal phase durations with an accuracy of more than $90\%$. Numerical search for the optimal phase durations is necessary since the optimization problem for phase duration is non-convex.
\subsection{Grid Search and Lookup Table}
In grid search, the optimal phase durations are obtained using exhaustive search over the entire range of $\alpha_1\geq 0, \alpha_2\geq 0,$ and $\alpha_1+\alpha_2\leq1$. The grid search is one dimensional $(\alpha_1\in[0,1])$ for the individual rate $(R_1)$ and two dimensional $(\alpha_1,\alpha_2)$ for the sum rate. Since the running time  for each set of phase durations is usually small, such grid search can be efficient for obtaining the accurate optimal phase duration. For practical implementation in channels that vary slightly, these optimal phase durations can be pre-computed offline for each set of channel gains, then stored in a table and the algorithm only needs
to perform table lookup at run time.

Fast varying channels may require  a large table to store the optimization results for every channel configuration. Alternatively, results for  a set of sampled channel gains can be stored. Then,  if the actual channel gain is in between stored values, a grid search for optimal phase durations can be performed in between
the two stored phase values instead of the whole range $[0,1],$ which significantly reduces the searching time.

For extension to the $m-$user case, the number of cases resulting from individual and sum rate optimizations is equal to $m+1$ and $2^{m}$, respectively. Numerical search for the optimal $\alpha_1^{\star}$ for the individual rate will consume similar time to that required for the two user case since it is still a one-dimensional search $(\alpha_2^{\star}=\alpha_3^{\star}=...=\alpha_{m-1}^{\star}=0)$. However, finding the optimal phases  $\alpha_1^{\star},\;\alpha_2^{\star},...,$
 and $\alpha_{m}^{\star}$ for the sum rate
 becomes an $m-$dimentional search. Next, we propose a simple method that can approximately achieves the optimal phase durations.

\subsection{Polynomial Approximation}\label{aprtec}
Although the optimization problem for phase durations is non-convex, it is observed through extensive numerical examples to have a unique maximum
 in the range $[0,1]$ (see Figure \ref{fig:conf} for examples of the individual and sum rates versus phase duration). In addition, the curves around the optimal phase durations
suggest that they can be approximated by quadratic functions. Therefore, we can use interpolation technique for the individual and the sum rates as follows.
\subsubsection{Individual Rate Interpolation}
After quantizing the interval of $\alpha_1$ $([0,1])$ into $L\geq3$ points, calculate the optimal power allocations for each quantized
value of $\alpha_1$. Choose the value $\alpha_{1,1}$ that leads to the maximum individual rate $R_{1,1}$ and two other
points $(\alpha_{1,2},R_{1,2})$ and $(\alpha_{1,3},R_{1,3})$ that directly surround $(\alpha_{1,1},R_{1,1})$.  Next, use these  points
to express $R_1$ as quadratic function of $\alpha_1$ \cite{numE}. The approximated optimal $\alpha_1^\diamond$ is obtained from the
derivative of this function.
\subsubsection{Sum Rate Interpolation}
Similarly, quantize the interval of $\alpha_1$ $([0,1])$ into $L\geq 3$ points. For each value of $\alpha_1$,
quantize the interval of $\alpha_2$ $([0,1-\alpha_1])$ into $T\geq3$ points. For each pair $(\alpha_1,\alpha_2)$,
calculate the optimal power allocations. Choose the point $(\alpha_{1,1},\alpha_{2,1}, S_{R,1})$ that has the maximum sum rate and
four other points that surround it. Use these five points
 to express $S_R$ as a bivariate quadratic function of $(\alpha_1,\alpha_2)$. The approximated optimal pair $(\alpha_1^\diamond,\alpha_2^\diamond)$  is then obtained from the derivatives of the interpolated function \cite{numE}.

 In Section \ref{sec:numresult},
numerical results comparing approximated and optimal phase durations
show that the interpolated values approximate  the exact values within a margin of $6\%$ for the individual rate and $10\%$ for the sum rate.
\section{Numerical Results}\label{sec:numresult}
In this section, we first compare the achievable rate region of our scheme with an outer bound in Figure \ref{fig:asyro} to show that it is near capacity-achieving. Then, in Figures \ref{fig:ratefh}--\ref{fig:ratsum}, we optimize our scheme for a network on a $2D$ plane with fixed user locations while letting the destination move on the plane and applying a channel gain model with pathloss only. These figures show the geometrical regions of optimal schemes for either the individual or the sum rate
 as well as the maximum rate. Figure \ref{fig:conf} illustrate the behavior of the maximum individual and sum rates versus phase durations. Next, Figure \ref{fig:phstar} illustrates the exact and approximate optimal phase durations for both the individual and sum rates obtained respectively by the  grid search and interpolation technique. Figure \ref{fig:ratsumlin} illustrates the gain in sum rate for symmetric channels. In all these figures, we normalize the transmission time and the channel bandwidth.

Figure \ref{fig:asyro} compares between the achievable rate regions of the proposed scheme, the classical MAC and an outer bound as derived in \cite{haivu3}. The outer bound consists of all rate pairs $(R_1,R_2)$ satisfying (\ref{th1Grrsimp}) but replacing $g_{12}^2$ by $g_{10}^2+g_{12}^2$ and $g_{21}^2$ by $g_{20}^2+g_{21}^2$. Results are plotted for different values of $g_{10}$ and $g_{20}$ while fixing $g_{12}=g_{21}=5$ and $P_1=P_2=2$. This figure shows that our scheme is near capacity-achieving: it is close to the outer bound especially when the ratios  $g_{12}/g_{10}$ and $g_{21}/g_{20}$ are high.
\noindent
\begin{figure}[t]
    \begin{center}
    \includegraphics[width=0.45\textwidth,  height=62mm]{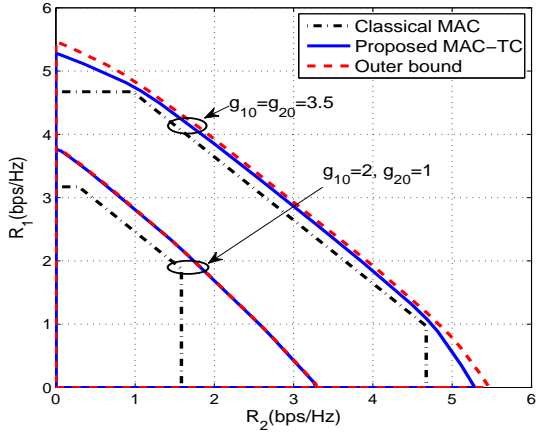}
    \caption{Achievable rate regions and outer bounds for the asymmetric half-duplex MAC-TC with $g_{12}=g_{21}=5$.}
    \label{fig:asyro}
    \end{center}
\vspace*{-5mm}
\end{figure}
Figure \ref{fig:ratefh} shows the geometrical regions of optimal schemes for maximizing  the individual rate as described in Theorem $2$. We fix $\alpha_1 = 0.5$,  fix the locations of user $1$ and user $2$ at points $(-0.5,0)$ and $(0.5,0)$, respectively, and allow the destination to be anywhere on the plane. We use a pathloss-only model in which  each channel gain $g_{ij}$ is related to the distance by $g_{ij}=d_{ij}^{-\gamma/2}$ where $\gamma=2.4$. Let $d_{12},$ $d_{10}$, and $d_{20}$  respectively be the distances between the two users, user $1$ and the destination, and user $2$ and the destination.  Results show that the optimal scheme is direct transmission if $d_{10}<d_{12}$, two-hop transmission if $d_{10}>2d_{12}$ and $d_{20}$ is slightly bigger than $d_{12}$, DF with message repetition if $d_{10}>>d_{12}$ and $d_{20}>d_{12}$, and PDF with or without message repetition in the remaining two regions.
Here we choose to fix $\alpha_1$ to simplify the computation, but similar results can also be obtained with the optimal $\alpha_1^{\star}$.
\noindent
\begin{figure}[t]
    \begin{center}
    \includegraphics[width=0.45\textwidth, height=62mm]{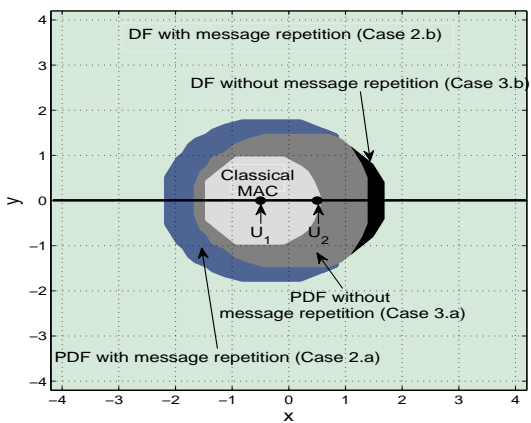}
    \caption{Geometric regions of optimal schemes for maximizing the individual rate per destination location with $\alpha_1=0.5$, $\gamma=2.4$, ($U_{1}=$ user $1$, $U_{2}=$ user $2$). The different cases correspond to those in Theorem 2.} \label{fig:ratefh}
    \end{center}
\vspace*{-5mm}
\end{figure}
\noindent
\begin{figure}[!t]
    \begin{center}
    \includegraphics[width=0.45\textwidth, height=62mm]{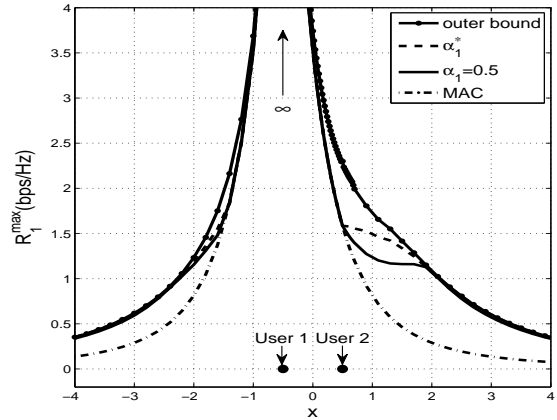}
    \caption{$R_1^{\max}$ with the optimal $\alpha_1^{\star}$ and with $\alpha_1=0.5$ when the destination moves on the x axis (solid horizontal line in Figure $3$).} \label{fig:raout}
    \end{center}
\vspace*{-5mm}
\end{figure}
\noindent
\begin{figure}[!t]
    \begin{center}
    \includegraphics[width=0.45\textwidth, height=62mm]{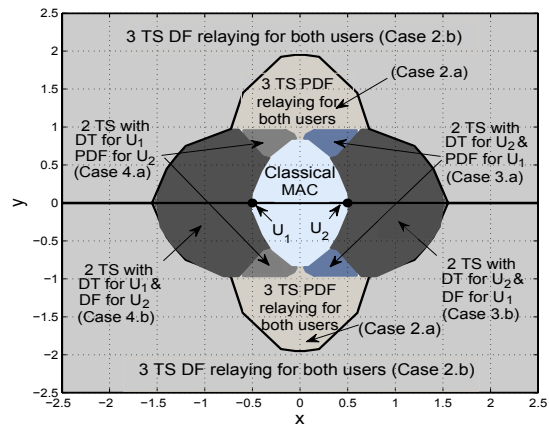}
    \caption{Geometric regions of optimal schemes for maximizing the sum rate per destination location with $\alpha_1=\alpha_2=0.2$, $\gamma=2.4$, ($U_{1}=$ user $1$, $U_{2}=$ user $2$, TS = time slot (phase), DT $=$ direct transmission, DF $=$ decode-forward, PDF $=$ partial DF). The different cases correspond to those in Theorem 3.} \label{fig:locsum}
    \end{center}
\vspace*{-6mm}
\end{figure}
Figure \ref{fig:raout} presents $R_1^{\max}$ versus distance as the destination moves along the line passing through both users for $\alpha_1=0.5$ and  for the optimal $\alpha_1^{\star}$, and compares this rate with the classical MAC and the outer bound.  Results show that as ratio $d_{10}/d_{12}$  increases, $R_1^{\max}$ becomes closer to the outer bound. This phenomenon is expected because when $d_{10}/d_{12}$ increases, $g_{12}/g_{10}$ increases such that $g_{12}^2\rightarrow g_{10}^2+g_{12}^2$. The two users then virtually become one entity and the channel approaches a single-user one with known capacity.

Figure \ref{fig:locsum} shows the regions of optimal schemes for maximizing  the sum rate at each destination location for $\alpha_1=\alpha_2=0.2$ and
the same channel configuration as in Figure \ref{fig:ratefh}. As the figure is symmetric, lets consider the right half plane. There are $5$
different regions that correspond to the first $5$ cases described in Theorem $3$. Results show that the optimal scheme is classical MAC if $d_{10}<d_{12}$ and $d_{20}<d_{12}$, $3$-phase scheme (either DF or PDF) if $d_{10}>d_{12}$ and $d_{20}>d_{12}$, and $2$-phase scheme (either DF or PDF) if $d_{10}>d_{12}$ and $d_{20}<d_{12}$. Furthermore, the optimal scheme switches from DF to PDF as the difference between $d_{10}$ and $d_{20}$ decreases. Similar results can be obtained with the optimal $\alpha_1^{\star}$ and $\alpha_2^{\star}$.

Considering the inverse relation between distance and channel gain in the pathloss model, Figures \ref{fig:ratefh} and \ref{fig:locsum} imply that
as the inter-user link qualities increase in relation to the user-destination link qualities,
the optimal scheme transverses
from no cooperation to partial then to full cooperation.     %
\noindent
\begin{figure}[!t]
    \begin{center}
    \includegraphics[width=0.45\textwidth, height=62mm]{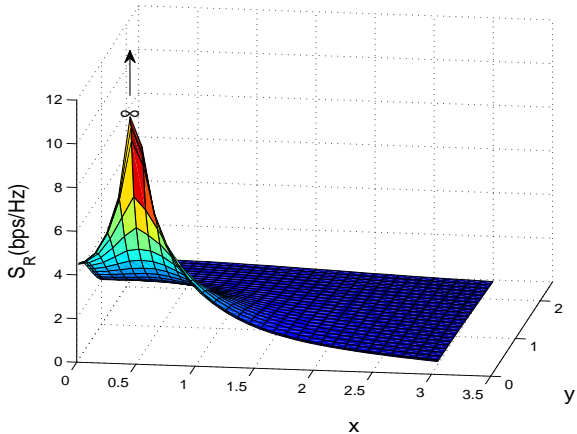}
    \caption{Maximum sum rate per destination location  with $\alpha_1=\alpha_2=0.2$ and $\gamma=2.4$.} \label{fig:ratsum}
    \end{center}
\vspace*{-6mm}
\end{figure}
\noindent
\begin{figure}[!t]
    \begin{center}
    \includegraphics[width=90mm]{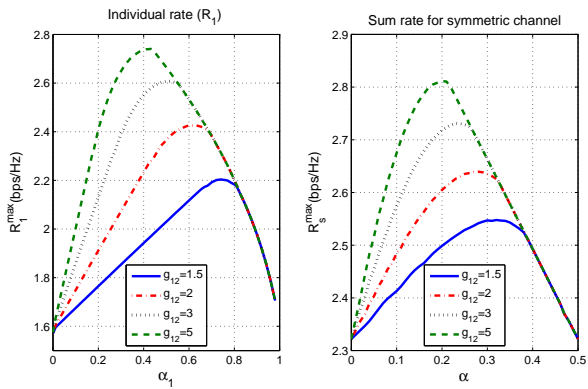}
    \caption{Maximum individual and sum rates versus $\alpha$ for symmetric channels with different values of $g_{12}$.}
    \label{fig:conf}
    \end{center}
\vspace*{-6mm}
\end{figure}
Figure \ref{fig:ratsum} shows the maximum sum rate versus distance when the destination moves in the first quadrant of Figure \ref{fig:locsum}
(other quadrants are symmetric to this one).
 Results show that as the destination moves closer to one of the users, the sum rate increases
 because the link quality between that user and the destination becomes high such that the user can send a large
 amount of information (as $d_{10}$ or $d_{20} \rightarrow 0$, $S_R^{\max}\rightarrow \infty$ bps/Hz).

Figure \ref{fig:conf} shows the maximum individual and sum rates versus $\alpha$ for symmetric channels. Results show that for each inter-user link quality, there is a unique phase duration that maximizes the individual or sum rate. Figure \ref{fig:phstar} shows the interpolated and exact $\alpha^{\star}$ for both the individual and the sum rates versus $g_{12}$ for symmetric channels.
Results show that $\alpha^{\star}$ decreases as $g_{12}$ increases, i.e.
the two users can exchange their information in a smaller portion of time and spend a bigger portion in cooperation. For the interpolated individual rate, the interval of $\alpha_1$ is quantized uniformly into $8$ points including $0$ and $1$ at which the rate is the same as direct transmission. Similarly for the sum rate where $\alpha_1=\alpha_2=\alpha$ since the channel is symmetric. Results show that the interpolated values are close to the optimal values with error less than $6\%$ for the individual rate and $10\%$ for the sum rate.

Figure \ref{fig:ratsumlin} compares the sum rate  of the MAC-TC with the classical MAC for symmetric channels. Results show that once $g_{12}>g_{10}$, the sum rate of the
MAC-TC starts increasing significantly and then reaches saturation. From (\ref{gain2}), the maximum sum rate gains that can be
obtained with $P=2,$ $4,$ and $10$ are $0.85,$ $0.92,$ and $0.95$ bps/Hz, respectively, while the asymptotic gain from Theorem $4$ (with $P\rightarrow \infty$) is
$1$ bps/Hz. Hence, the sum rate gain approaches the asymptotic gain even at low transmission power and finite link gain.
\noindent
\begin{figure}[!t]
    \begin{center}
    \includegraphics[width=0.45\textwidth, height=62mm]{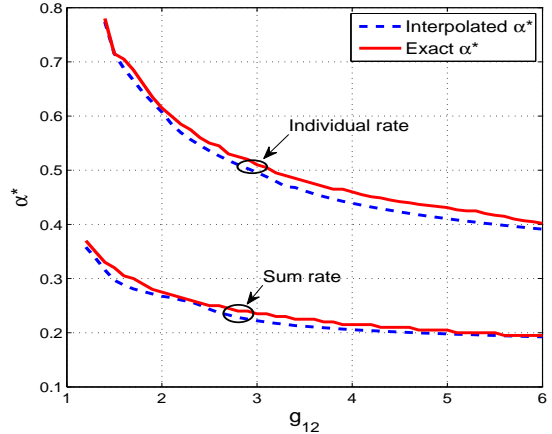}
    \caption{Interpolated and exact optimum $\alpha^{\star}$ versus $g_{12}$ for the individual and sum rates of symmetric channels ($g_{10}=1$).} \label{fig:phstar}
    \end{center}
\vspace*{-6mm}
\end{figure}
\noindent
\begin{figure}[!t]
    \begin{center}
    \includegraphics[width=0.45\textwidth, height=62mm]{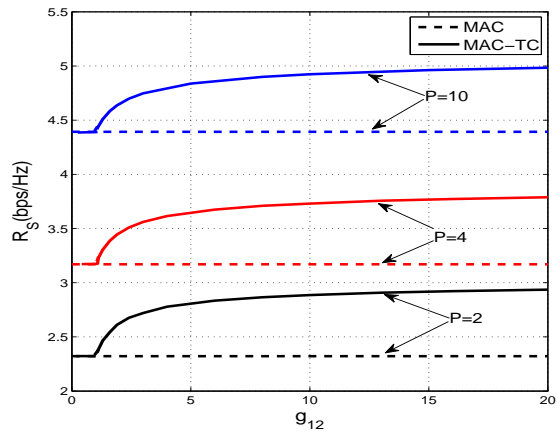}
    \caption{Maximum sum rate with the optimal $\alpha^{\star}$ for the MAC-TC and classical MAC versus $g_{12}$ for symmetric channels with $g_{10}=1$.} \label{fig:ratsumlin}
    \end{center}
\vspace*{-6mm}
\end{figure}
\section{Conclusion}\label{sec:concl}
We have proposed a half-duplex cooperative scheme for the MAC-TC and optimized its resource allocation to obtain the maximum individual or sum rate performance. With fixed phase durations, the power allocation problem is convex and can be solved analytically based on the KKT conditions. The optimal phase duration can then be found by numerical search such as a simple grid search, lookup table or quadratic interpolation.
Depending on channel
conditions, the optimal scheme can be direct transmission, DF or PDF and with or without
 message repetition, or multi-hop forwarding. We also present numerical results that illustrate the optimal scheme for each geometrical location of the destination on a plane while fixing user
 locations. These analyses provide  guidelines for implementing
  the proposed scheme in an actual network.

  As future works, it would
of interest to consider the fading channel and study its effects
on the optimal resource allocation for maximizing the average individual and sum rates or minimizing the outage
probability.
\appendices
\section{Proof of Corollary $1$}
The PDF scheme is similar to the proposed scheme in Section \ref{sec:copshrr}, but each user transmits both message parts in the first two phases and decodes only one. Therefore, both users construct their transmitted signals as in (\ref{sigtr}) but with $X_{11}$ and $X_{22}$ as
\begin{align*}
X_{11}&=\sqrt{\rho_{11}}U_{12}(w_{12})+\sqrt{\rho_{10}^{\dag}}U_{10}(w_{10}),\nonumber\\
X_{22}&=\sqrt{\rho_{22}}U_{21}(w_{21})+\sqrt{\rho_{20}^{\dag}}U_{20}(w_{20})
\end{align*}
\noindent where $U_{10},U_{20}, V_{10},V_{20},U, U_{11},U_{22}$ are i.i.d $\sim N(0,1)$
and the power constraints are
\noindent
\begin{align}\label{powcsch222}
\alpha_1(\rho_{11}+\rho_{10}^{\dag})+\alpha_3(\rho_{10}+\rho_{13})&=P_1,\nonumber\\
\alpha_2(\rho_{22}+\rho_{20}^{\dag})+\alpha_3(\rho_{20}+\rho_{23})&=P_2.
\end{align}
The achievable rate for this scheme can be derived using a similar procedure as in \cite{haivu3} and
consists of rate pairs $(R_1,R_2)$ satisfying (\ref{th1Grrsimp}) with the following modifications. In $J_1,$ $J_2,$ $S_1$, replace
\begin{align}\label{difter}
&C(g_{12}^2\rho_{11})\;\text{by}\;{\cal C}\!\!\left(\frac{g_{12}^2\rho_{11}}{1+g_{12}^2\rho_{10}^{\dag}}\right)
+{\cal C}\!\!\left(g_{10}^2\rho_{10}^{\dag}\right)\;\text{and}\nonumber\\
&C(g_{21}^2\rho_{22})\;\text{by}\;{\cal C}\!\!\left(\frac{g_{21}^2\rho_{22}}{1+g_{21}^2\rho_{20}^{\dag}}\right)
+{\cal C}\!\!\left(g_{20}^2\rho_{20}^{\dag}\right)
\end{align}
\noindent  and in $S_2,$ $S_3,$ and $S_4$, replace
\begin{align*}
&g_{10}^2\rho_{11}\;\text{by}\;g_{10}^2(\rho_{11}+\rho_{10}^{\dag}),\;g_{10}^2\rho_{22}\;\text{by}\;g_{20}^2(\rho_{22}+\rho_{20}^{\dag}),\nonumber\\
&g_{12}^2\rho_{11}\;
\text{by}\;\frac{g_{12}^2\rho_{11}}{1+g_{12}^2\rho_{10}^{\dag}},\;\text{and}\;
g_{21}^2\rho_{22}\;\text{by}\;
\frac{g_{21}^2\rho_{22}}{1+g_{21}^2\rho_{20}^{\dag}}.
\end{align*}
 Lets denote the rate constraints for this scheme as $(J_1^{\dag},J_2^{\dag},...,S_4^{\dag})$.
This new region is equivalent to the region in (\ref{th1Grrsimp}) if the optimal $\rho_{10}^{\dag}$ and $\rho_{20}^{\dag}$ are zero.

First, the sum rates $S_2^{\dag}$, $S_3^{\dag}$ and $S_4^{\dag}$ can be maximized with $\rho_{10}^{\dag}=\rho_{20}^{\dag}=0$. With these values, the second terms in $S_2^{\dag}$ and $S_3^{\dag}$ are maximized while the first terms are not affected because we can allocate the whole power in the $1^{\text{st}}$ and $2^{\text{nd}}$ phases to $\rho_{11}$ and $\rho_{22}$, respectively.

Second, for $J_1^{\dag}$, $J_2^{\dag}$ and $S_1^{\dag}$, the different terms from $J_1$, $J_2$ and $S_1$
are given by the new terms in (\ref{difter}). If $g_{12}>g_{10}$ and $g_{21}>g_{20}$,
these terms are maximized with $\rho_{10}^{\dag}=\rho_{20}^{\dag}=0$. This can be shown as follows. Let $P_1^{(1)}$ be the optimal power transmitted from the first user in the $1^{\text{st}}$ phase. Then, the first new term in (\ref{difter}) can be  maximized using the Lagrangian
\begin{align}\label{T1opt}
&L(\rho_{11},\rho_{10}^{\dag},\lambda)=\alpha_1{\cal C}\!\left((1+g_{12}^2\rho_{10}^{\dag})^{-1}g_{12}^2\rho_{11}\right)\nonumber\\
&+\alpha_1{\cal C}\!\left(g_{10}^2\rho_{10}^{\dag}\right)
+\lambda(P_1^{(1)}-\alpha_1(\rho_{11}+\rho_{10}^{\dag})).
\end{align}
\noindent Taking derivatives with respect to $\rho_{11}$ and $\rho_{10}^{\dag}$, we obtain
\begin{align*}
\frac{\partial L}{\partial \rho_{11}}&=\frac{g_{12}^2}{1+g_{12}^2(\rho_{11}+\rho_{10}^{\dag})}-\alpha_1\lambda,\nonumber\\
\frac{\partial L}{\partial \rho_{10}^{\dag}}&=\frac{-g_{12}^4\rho_{11}}{(1+g_{12}^2\rho_{10}^{\dag})(1+g_{12}^2(\rho_{11}+\rho_{10}^{\dag}))}\nonumber\\
&\;\;+\frac{g_{10}^2}{1+g_{10}^2\rho_{10}^{\dag}}-\alpha_1\lambda.
\end{align*}
\noindent Next, obtain $\lambda$ from setting $\frac{\partial L}{\partial \rho_{11}}=0$ and substitute it
into $\frac{\partial L}{\partial \rho_{10}^{\dag}}$ to get
\begin{align*}
\partial L/\partial
\rho_{10}^{\dag}=\alpha_1(g_{10}^{-2}+\rho_{10}^{\dag})^{-1}-\alpha_1(g_{11}^{-2}+\rho_{10}^{\dag}).
\end{align*}
\noindent Since $g_{12}>g_{10}$, $\frac{\partial L}{\partial \rho_{10}^{\dag}}<0$, this first new term in (\ref{difter}) is decreasing in $\rho_{10}^{\dag}$ and the optimal $\rho_{10}^{\dag}=0$. Similarly, the optimal value  $\rho_{20}^{\dag}=0$.

Third, if $g_{12}<g_{10}$ and $g_{21}<g_{20}$, decoding at any user will limit the rates; hence, we have $\alpha_1^{\star}=\alpha_2^{\star}=0.$ Then, both coding schemes reduces to the classical MAC.

Fourth, if $g_{12}>g_{10}$ and $g_{21}<g_{20}$, decoding at user $1$ will limit the rate; hence, the optimal scheme is obtained by setting $\alpha_2=0$. Then, we can show that the optimal $\rho_{10}^{\dag}=0$ following the same procedure in the first and second steps.
Similarly for  $g_{12}<g_{10}$ and $g_{21}>g_{20}$.
\section{Proof of the Algorithm for the Optimal Individual Rate}
In optimization problem (\ref{Eq:bed5}), we first assume that $\rho_{i}^{\star}>0$ for $I\in\{10,12,13\}$. Hence, their constraints are inactive and
$\underline{\mu}^{\star}=0$. Therefore, the Lagrangian function in (\ref{lagin}) becomes
\begin{align}\label{derinf}
&L(R_1,\underline{\rho},\underline{\lambda})=
R_1 -\lambda_0(R_1-J_1)-\lambda_1(R_1-S_4)\nonumber\\
&\;-\lambda_2\left(\alpha_1 \rho_{11}+(1-\alpha_1)(\rho_{10}+\rho_{13})-P\right).
\end{align}
\noindent Now, we verify the KKT conditions in (\ref{KKTind1})--(\ref{KKTind9}). Starting from (\ref{KKTind1}),  we take the derivatives
of $L(R_1,\underline{\rho},\underline{\lambda})$ in (\ref{derinf}) with respect to all variables as follows.
\begin{align}\label{derin2}
\frac{\partial L}{\partial R_1}&=1-\lambda_0-\lambda_1,\\
\frac{\partial L}{\partial \rho_{13}}&=\frac{(1-\alpha_1)\lambda_1\left(1+\frac{g_{20}}{g_{10}}\sqrt{\frac{\rho_{23}}{\rho_{13}}}\right)}
{\frac{1}{g_{10}^2}+\rho_{10}+\rho_{13}+\frac{g_{20}^2}{g_{10}^2}\rho_{23}+2\frac{g_{20}}{g_{10}}\sqrt{\rho_{13}\rho_{23}}}-(1-\alpha_1)
\lambda_2,\nonumber\\
\frac{\partial L}{\partial \rho_{11}}&=\frac{\alpha_1\lambda_0}{\frac{1}{g_{12}^2}+\rho_{11}}+\frac{\alpha_1 \lambda_1}{\frac{1}{g_{10}^2}+\rho_{11}}-\alpha_1\lambda_2\nonumber\\
\frac{\partial L}{\partial \rho_{10}}&=\frac{(1-\alpha_1)\lambda_0}{\frac{1}{g_{10}^2}+\rho_{10}}-(1-\alpha_1)\lambda_2\nonumber\\
&\;\;+\frac{(1-\alpha_1)\lambda_1}
{\frac{1}{g_{10}^2}+\rho_{10}+\rho_{13}+\frac{g_{20}^2}{g_{10}^2}\rho_{23}+2\frac{g_{20}}{g_{10}}\sqrt{\rho_{13}\rho_{23}}}.
\end{align}
\noindent First, by setting $\frac{\partial L}{\partial R_1}=0$ and $\frac{\partial L}{\partial \rho_{13}}=0$, we obtain formulas for $\lambda_0$ and $\lambda_2$.  By substituting these formulas into the third and fourth equations in (\ref{derin2}) and equating them to zero, we get the following equations for $\lambda_1$
\begin{align}\label{derin3}
\lambda_1&=\frac{\xi}{\left(1+\frac{g_{20}}{g_{10}}\sqrt{\frac{\rho_{23}}{\rho_{13}}}\right)\left(\frac{1}{g_{12}^2}+\rho_{11}\right)+\xi-
\xi\frac{g_{12}^{-2}+\rho_{11}}{g_{10}^{-2}+\rho_{11}}},\nonumber\\
\lambda_1&=\frac{\xi}{\xi+\frac{g_{20}}{g_{10}}\sqrt{\frac{\rho_{23}}{\rho_{13}}}
\left(\frac{1}{g_{10}^2}+\rho_{10}\right)}
\end{align}
\noindent where $\xi=\frac{1}{g_{10}^2}+\rho_{10}+\rho_{13}+\frac{g_{20}^2}{g_{10}^2}\rho_{23}+2\frac{g_{20}}{g_{10}}\sqrt{\rho_{13}\rho_{23}}$. By equalizing  these two equations, we obtain $\rho_{10}^{\star}$ as in Table I (Case$2.a$).
Then, by substituting these expressions into the power constraint in (\ref{Eq:bed5}) such that KKT condition (\ref{KKTind2}) is satisfied, we get
\begin{align}\label{derin5}
P&=\alpha_1 \rho_{11}^{\star}+(1-\alpha_1)(\rho_{10}^{\star}+\rho_{13}^{\star})\nonumber\\
&=\alpha_1 (a_2-g_{10}^{-2})+(1-\alpha_1)(\rho_{10}^{\star}+\rho_{13}^{\star}).
\end{align}
\noindent By substituting $\rho_{10}^{\star}$ into (\ref{derin5}), we get $\rho_{11}^{\star}$ in Table I.
Hence, we find $\rho_{11}^{\star}$ in terms of $\rho_{13}^{\star}$. In order to find $\rho_{10}^{\star}$, we use KKT conditions (\ref{KKTind2})--(\ref{KKTind4}). While (\ref{KKTind2}) is the power constraint, KKT conditions (\ref{KKTind3}) and (\ref{KKTind4})
are equivalent to $J_1=S_4$ when both constraints are tight. By using the power constraint and the equality $J_1=S_4$, we obtain $f_2(\rho_{13}^{\star})=0$ in (\ref{def1}). Therefore, for any $\alpha_1$
we can find  $\rho_{13}^{\star}$, $\rho_{10}^{\star},\rho_{11}^{\star}$ and $\rho_{13}^{\star}$ from the above equations if there is a solution
for $f_2(\rho_{13}^{\star})=0$ in (\ref{def1}) that gives $\rho_{10}\geq0$. This corresponds to  Case $2.a$.
\subsubsection*{Case $2.b$}
If $f_2(\rho_{13}^{\star})$ in (\ref{def1}) has a solution that gives $\rho_{10}<0$, KKT condition (\ref{KKTind5}) is not satisfied. Therefore, to satisfy KKT conditions (\ref{KKTind5}) and (\ref{KKTind7}), we have $\mu_{10}^{\star}> 0$ and $\rho_{10}^{\star}=0$. Hence, $\rho_{10}^{\star}=0$ is an active constraint. Then,  we consider (\ref{derinf}) with $\rho_{10}^{\star}=0$. By using the power constraint (KKT condition (\ref{KKTind2})) and solving $J_1=S_4$ (KKT conditions (\ref{KKTind3}) and (\ref{KKTind4})), we can find
  $\rho_{11}^{\star}$ and $\rho_{13}^{\star}$ as in Table I (Case $2.b$).
\subsubsection*{Case $3$}
As explained in Section \ref{sec:indivr}, if there is no solution for $f_2(\rho_{13}^{\star})=0$ in (\ref{def1}), this means that $J_1 < S_4$,
i.e. only (\ref{KKTind3}) is tight while (\ref{KKTind4}) is not tight. From KKT condition (\ref{KKTind9}), $\lambda_1^{\star}=0$. Then, the
problem becomes as in (\ref{derinf}) with $\lambda_1^{\star}=0$.
 Immediately, $\rho_{13}^{\star}=0$ $(\mu_{13}>0)$ because it doesn't affect $J_1$. From the power constraint, we obtain $\rho_{10}^{\star}=(1-\alpha_1)^{-1}(P_1-\alpha_{13}\rho_{11})$. By substituting this $\rho_{10}^{\star}$ into $J_1$ and equalizing to
 zero the derivative of  $J_1$ with respect to $\rho_{11}$, we obtain
$\rho_{11}^{\star}$ in Case $3.a$ or $3.b$
\section{Proof of the Algorithm for the Optimal Sum Rate}
\subsection{Case $2$: $g_{12}>g_{10}$ and $g_{21}>g_{20}$}\label{prvsec}
In this case, $S_4<S_2$ and $S_4< S_3$. Therefore, the constraints $S_R-S_2<0$ and $S_R-S_3<0$ in KKT condition (\ref{KKTsum4}) are inactive. Moreover,  we first assume
that $\rho_{i}^{\star}>0$ for $I\in\{10,20,12,21,13,23\}$. Hence, their constraints in KKT condition (\ref{KKTsum5}) are inactive. Therefore, from KKT conditions
 (\ref{KKTsum7}) and (\ref{KKTsum8}), we have $\lambda_1=\lambda_2=0$ and $\mu_i^{\star}=0$. Then, the Lagrangian function in (\ref{lagsum}) becomes
\begin{align}\label{dersum3}
L(S_R,\underline{\rho},\underline{\lambda})=&S_R-\lambda_0(S_R-S_1)-\lambda_3\big(S_R-S_4\big)\\
&-\lambda_4\left(\alpha_{13} \rho_{11}+\alpha_3(\rho_{10}+\rho_{13})-P_1\right)\nonumber\\
&-\lambda_5\left(\alpha_2 \rho_{22}+\alpha_3(\rho_{20}+\rho_{23})-P_2\right).\nonumber
\end{align}
\noindent From KKT condition (\ref{KKTsum1}), we take the derivatives over all variables and get
\begin{align}\label{dersum4}
\frac{\partial L}{\partial S_R}&=1-\lambda_0-\lambda_3,\\
\frac{\partial L}{\partial \rho_{13}}&=\frac{\alpha_3\lambda_3g_{10}\left(g_{10}+g_{20}\sqrt{\frac{\rho_{23}}{\rho_{13}}}\right)}
{\xi_2}-\alpha_3\lambda_4,\nonumber\\
\frac{\partial L}{\partial \rho_{23}}&=\frac{\alpha_3\lambda_3g_{20}\left(g_{20}+g_{10}\sqrt{\frac{\rho_{13}}{\rho_{23}}}\right)}
{\xi_2}-\alpha_3\lambda_5\nonumber
\end{align}
\begin{align}
\frac{\partial L}{\partial \rho_{10}}&=\frac{\alpha_3 g_{10}^2\lambda_0}{1+g_{10}^2\rho_{10}+g_{20}^2\rho_{20}}+
\frac{g_{10}^2\alpha_3\lambda_3}{\xi_2}
-\alpha_3\lambda_4,\nonumber\\
\frac{\partial L}{\partial \rho_{20}}&=\frac{\alpha_3 g_{20}^2\lambda_0}{1+g_{10}^2\rho_{10}+g_{20}^2\rho_{20}}+
\frac{g_{20}^2\alpha_3\lambda_3}{\xi_2}
-\alpha_3\lambda_5\nonumber\\
\frac{\partial L}{\partial \rho_{11}}&=\frac{\alpha_1 g_{12}^2\lambda_0}{1+g_{12}^2\rho_{11}}+
\frac{g_{10}^2\alpha_1\lambda_3}{1+g_{10}^2\rho_{11}}-\alpha_1\lambda_4,\nonumber\\
\frac{\partial L}{\partial \rho_{22}}&=\frac{\alpha_2 g_{21}^2\lambda_0}{1+g_{21}^2\rho_{22}}+
\frac{g_{20}^2\alpha_2\lambda_3}{1+g_{20}^2\rho_{22}}-\alpha_2\lambda_5\nonumber
\end{align}
\noindent where $\xi_2=1+g_{10}^2(\rho_{10}+\rho_{13})+g_{20}^2(\rho_{20}+\rho_{23})+2g_{10}g_{20}\sqrt{\rho_{13}\rho_{23}}$. By setting $\frac{\partial L}{\partial S_R}=0,$ $\frac{\partial L}{\partial \rho_{13}}=0$ and $\frac{\partial L}{\partial \rho_{23}}=0$, we obtain formulas for $\lambda_0,\;\lambda_4$ and $\lambda_5,$ respectively.
By substituting these formulas into the other $4$ equations and setting them to $0$, we get the following equations for $\lambda_3$:
\begin{align}\label{dersum5}
\!\!\!\lambda_3^{(1)}&=\frac{\xi_2}{\xi_2\!+g_{10}g_{20}\sqrt{\rho_{23}\rho_{13}^{-1}}
\left(g_{10}^{-2}\!+\!\rho_{10}\!+g_{20}^2g_{10}^{-2}\rho_{20}\right)},\\
\!\!\!\lambda_3^{(2)}&=\frac{\xi_2}{\xi_2\!+g_{10}g_{20}\sqrt{\rho_{13}\rho_{23}^{-1}}
\left(g_{20}^{-2}\!+g_{10}^2g_{20}^{-2}\rho_{10}\!+\!\rho_{20}\right)},\nonumber\\
\!\!\!\lambda_3^{(3)}&=\frac{\xi_2}{g_{10}\left(g_{10}+g_{20}\sqrt{\rho_{23}\rho_{13}^{-1}}\right)\left(g_{12}^{-2}+\rho_{11}\right)
+\xi_2-\xi_2\left(\frac{g_{12}^{-2}+\rho_{11}}{g_{10}^{-2}+\rho_{11}}\right)},\nonumber\\
\!\!\!\lambda_3^{(4)}&=\frac{\xi_2}{g_{20}\left(g_{20}+g_{10}\sqrt{\rho_{13}\rho_{23}^{-1}}\right)\left(g_{21}^{-2}+\rho_{22}\right)
+\xi_2-\xi_2\left(\frac{g_{21}^{-2}+\rho_{22}}{g_{20}^{-2}+\rho_{22}}\right)}\nonumber
\end{align}
\noindent By comparing these equations, we get
\begin{align}\label{dersum6}
\lambda_3^{(1)}&=\lambda_3^{(2)}\Rightarrow\;\rho_{23}=\frac{g_{10}^2}{g_{20}^2}\rho_{13}\\
&\Rightarrow\; \xi_2=1+g_{10}^2\rho_{10}+g_{20}^2\rho_{20}+4g_{10}^2\rho_{13}\nonumber\\
\lambda_3^{(1)}&=\lambda_3^{(3)}\Rightarrow\;a_6=\frac{a_3-a_1}{2a_3-a_1}(2g_{10}^2a_3-G),\nonumber\\
\lambda_3^{(2)}&=\lambda_3^{(4)}\Rightarrow\;a_6=\frac{a_4-a_2}{2a_4-a_2}(2g_{20}^2a_4-G)\nonumber
\end{align}
\noindent where $G=4g_{10}^2\rho_{13}$,
and $a_6,\;a_1,\;a_2,\;a_3$ and $a_4$ are defined in (\ref{def2}).
By getting two formulas for $G$ from (\ref{dersum6}) and equalizing them, we get  $\rho_{22}^{\star}$ in Table II (Case $2.$a).

Up to this point, we have found the relations between $\rho_{13}$ and $\rho_{23}$, and between $\rho_{11}$ and $\rho_{22}$. To find the relation between $\rho_{11}$ and $a_6$ directly, we use the power constraints in (\ref{Eq:bed5}) to satisfy KKT conditions (\ref{KKTsum2}) and (\ref{KKTsum3}). Then, we get
\begin{align}\label{dersum10}
g_{10}^2\rho_{10}+g_{20}^2\rho_{20}&=
\alpha_3^{-1}\big(g_{10}^2(P_1-\alpha_1\rho_{11})\big)
-g_{10}^2\rho_{13}\nonumber\\
&\;\;+\alpha_3^{-1}\big(g_{20}^2(P_2-\alpha_2\rho_{22})\big)-g_{20}^2\rho_{23}\nonumber\\
\rightarrow a_6-1&=
\alpha_3^{-1}\big(g_{10}^2P_1+g_{20}^2P_2-\alpha_1g_{10}^2\rho_{11}-\alpha_2g_{20}^2\rho_{22}\big)\nonumber\\
&\;\;-0.5G
\end{align}
\noindent where (\ref{dersum10}) follows from $g_{10}^2\rho_{13}=g_{20}^2\rho_{23}$. By substituting $\rho_{11}$, $\rho_{22}$ and $G$ in terms of $a_3$ and $a_4$,
 we obtain $f_3(\rho_{11}^{\star})$ in (\ref{def2}).
Hence, for each $a_6$, we find $\rho_{23}^{\star},$ $\rho_{11}^{\star}$ and $\rho_{22}^{\star}$. Then, we find two expressions for $\rho_{13}$
from $G$ and $S_1=S_4$ which resembles KKT condition (\ref{KKTsum4}); subtract them to get $f_2(a_6)$ in (\ref{def2}).
Finally, from the power constraints, we obtain $\rho_{10}^{\star}$ and $\rho_{20}^{\star}$ in Table II (Case $2.a$).
If the solution for $f_2(a_6^{\star})$ in (\ref{def2}) gives positive $a_6^{\star}-1$, the results correspond to Case $2.a$.
\subsubsection*{Case $2.b$} $a_6^{\star}-1<0$ as in the solution for $f_2(a_6^{\star})$ in (\ref{def2}).

This case means that KKT condition (\ref{KKTsum5}) is not satisfied. To satisfy  KKT conditions (\ref{KKTind5}) and (\ref{KKTsum6}), we have $\mu_{10}^{\star},\;\mu_{20}^{\star}> 0$ and $\rho_{10}^{\star}=\rho_{20}^{\star}=0$ and they are active constraints. Therefore, we obtain the Lagrangian function
in (\ref{dersum3}) but with $\rho_{10}=\rho_{20}=0$.
First, from KKT conditions (\ref{KKTsum2}) and (\ref{KKTsum3}) (power constraints in (\ref{Eq:bed5})), we obtain the relation between $\rho_{11}^{\star}$ and $\rho_{13}^{\star}$; and between $\rho_{22}^{\star}$ and $\rho_{23}^{\star}$ as in Table II (Case $2.b$). Then, by taking the derivatives, we obtain $\frac{\partial L}{\partial \rho_{13}}$,
$\frac{\partial L}{\partial \rho_{23}}$, $\frac{\partial L}{\partial \rho_{11}}$, and $\frac{\partial L}{\partial \rho_{22}}$ as in (\ref{dersum4})
except with $\xi_2$ replaced by $\bar{\xi}_2=1+(g_{10}\sqrt{\rho_{13}}+g_{20}\sqrt{\rho_{23}})^2.$
Following  similar steps to the previous derivation, we get $\lambda_3^{(3)}$ and $\lambda_3^{(4)}$ in (\ref{dersum5}).
 By equalizing these two $\lambda$, we obtain $f_4(\rho_{23}^{\star})$ as in (\ref{def2}). Finally, from KKT condition (\ref{KKTsum4}) $(S_1=S_4)$, we obtain
  $f_5(\rho_{13}^{\star})$ as in (\ref{def2}).
\vspace{-3 mm}
\subsection{Case $3$: $g_{12}>g_{10}$ and $g_{21}\leq g_{20}$}
For this case, we use the Lagrangian function in (\ref{dersum3}) but
with $\alpha_2=0,$ $\rho_{22}=0$.
The derivatives of the Lagrangian function with respect to all variables are similar to those in (\ref{dersum4}) but without $\frac{\partial L}{\partial \rho_{22}}$.
As in Appendix \ref{prvsec}, by substituting $\lambda_4$ and $\lambda_5$ given in the first $2$ equations in (\ref{dersum4}) into the next $3$ equations and
setting them to $0$, we get the same  equations for $\lambda_3$ in (\ref{dersum5}) but without $\lambda_3^{(4)}$. Since the first three equations in (\ref{dersum5}) are equal, by comparing them, we get the
first two equations in (\ref{dersum6}). By using the power constraint (KKT conditions (\ref{KKTsum2}) and (\ref{KKTsum3})), we get the same formula in (\ref{dersum10}) but with $\rho_{22}=0$ and $\alpha_2=0$, then we can obtain the relation between $\rho_{11}^{\star}$ and
$\rho_{10}^{\star}$ as given in Table II (Case $3.a$).

We now have $4g_{10}^2\rho_{13}=G$ and $S_1=S_4$ (KKT condition (\ref{KKTsum4})). By subtracting these equations, we get $f_7(\rho_{10}^{\star})$  in (\ref{def2}). Finally, we obtain $\rho_{20}^{\star}$, $\rho_{13}^{\star}$ and $\rho_{23}^{\star}$  in Table II (Case $3.a$).
 If the solution for $f_7(\rho_{10}^{\star})$ in (\ref{def2}) gives $\rho_{10}^{\star}>0$, we have Case $3.a$.
\subsubsection*{Case $3.b$} $\rho_{10}^{\star}<0$ as in the solution for $f_7(\rho_{10}^{\star})$ in (\ref{def2}).

In this case, our assumption that $\underline{\rho}>0$ is inactive is not correct. In order to satisfy KKT condition (\ref{KKTsum5}), we have $\rho_{10}^{\star}=0$ and  $\mu_{10}>0$. Then, the Lagrangian function in (\ref{dersum3}) is considered with $\rho_{10}=\rho_{22}=0$
and $\alpha_2=0$. From the power constraints (KKT conditions (\ref{KKTsum2}) and (\ref{KKTsum3})), we obtain $\rho_{20}^{\star}$ and $\rho_{11}^{\star}$
 in Table II (Case $3.a$).  Then, by taking the derivatives, we obtain
$\frac{\partial L}{\partial \rho_{13}}$,
$\frac{\partial L}{\partial \rho_{23}}$, $\frac{\partial L}{\partial \rho_{20}}$
as in (\ref{dersum4}) but with $\xi_2$ replaced by $\xi_3=1+g_{20}^2\rho_{20}+(g_{10}\sqrt{\rho_{13}}+g_{20}\sqrt{\rho_{23}})^2.$
\noindent Following similar steps to  Case $2.a$, we obtain
$\lambda_1^{(2)}$ and $\lambda_1^{(3)}$ in (\ref{dersum5}) but with $\xi_3$
instead of $\xi_2$.
By setting $\lambda_3^{(2)}=\lambda_3^{(3)}$, we obtain $\rho_{23}^{\star}$ in Table II (Case $3.a$)..
Then, from $S_1=S_4$ (KKT condition (\ref{KKTsum4})),
we obtain $f_8(\rho_{13}^{\star})$ in (\ref{def2}).
\bibliographystyle{IEEEtran}
\bibliography{references}
\begin{IEEEbiography}[{\includegraphics[width=1in,height=1.25in,clip,keepaspectratio]{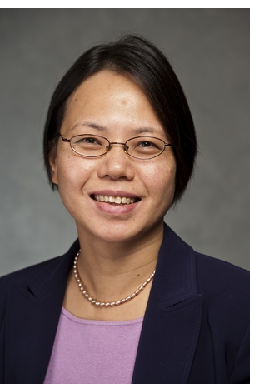}}]{Mai Vu}
received a PhD degree in Electrical Engineering from Stanford University after having an
MSE degree in Electrical Engineering from the
University of Melbourne and a bachelor degree in
Computer Systems Engineering from RMIT, Australia. Between $2006-2008$, she worked as a lecturer
and researcher at the School of Engineering and
Applied Sciences, Harvard University. During $2009-
2012$, she was an assistant professor in Electrical and
Computer Engineering at McGill University. Since
January $2013$, she has been an associate professor in
the department of Electrical and Computer Engineering at Tufts University.
Dr. Vu conducts research in the general areas of wireless communications,
signal processing for communications, network communications and information theory. Examples include cooperative and cognitive communications,
relay networks, MIMO systems. Dr. Vu has served on the technical program
committee of numerous IEEE conferences and is currently an editor for the
IEEE Transactions on Wireless Communications.
\end{IEEEbiography}
\begin{IEEEbiography}[{\includegraphics[width=1in,height=1.25in,clip,keepaspectratio]{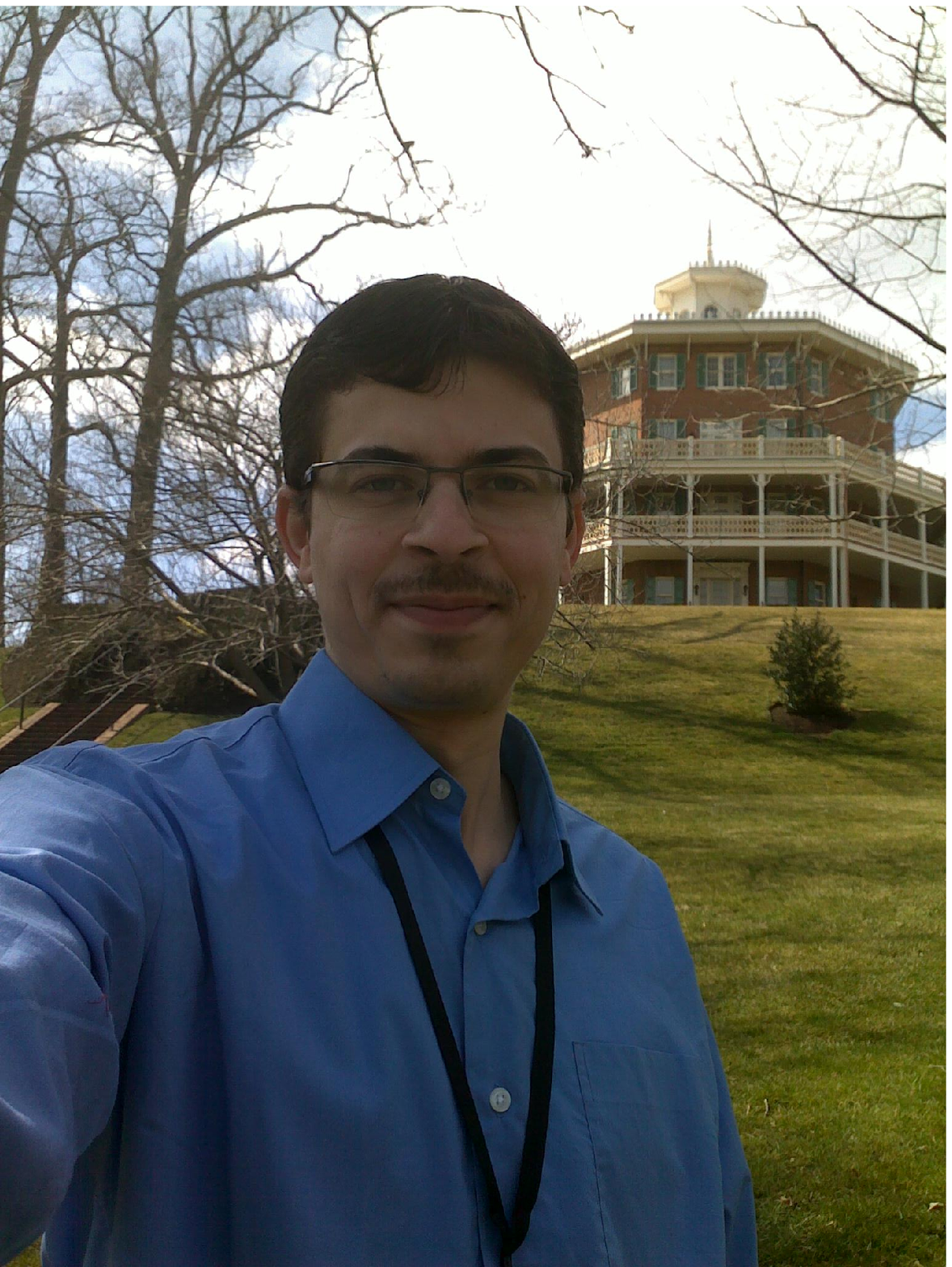}}]{Ahmad Abu Al Haija}
received the B.Sc. and M.Sc. degrees in Electrical Engineering from Jordan University of Science and Technology (JUST), Irbid, Jordan, in $2006$ and $2009$, respectively.
He is working toward the Ph.D. degree since January $2010$ at McGill University, Montreal, Canada. Since February $2013,$ he is a visiting student at Tufts University, Medford, MA.
His research interest includes cooperation in multiuser channels, resource allocation for cooperative communication systems, wireless communications and performance analysis evaluation over fading channels.
\end{IEEEbiography}
\end{document}